%% file: main.tex
\documentclass[lettersize,journal]{IEEEtran}
\IEEEoverridecommandlockouts
\pdfoutput=1
\usepackage{algorithm}
\usepackage{algpseudocode}
\usepackage{amsmath}
\usepackage{amsthm}
\usepackage{array}
\usepackage{graphicx}
\usepackage{subcaption}
\newcolumntype{C}[1]{>{\centering\let\newline\\\arraybackslash\hspace{0pt}}m{#1}}
\newcolumntype{R}[1]{>{\raggedleft\let\newline\\\arraybackslash\hspace{0pt}}m{#1}}
\usepackage{xspace}
\usepackage{threeparttable}
\usepackage{color}
\usepackage{colortbl}
\usepackage{url}
\usepackage{paralist}
\usepackage{wrapfig}
\usepackage{multirow}
\usepackage{listings}
\usepackage{comment}
\usepackage{booktabs}
\usepackage{makecell}
\usepackage{boxedminipage}
\usepackage{amsmath}
\usepackage{graphicx}
\usepackage[numbers,sort&compress]{natbib}
\usepackage{minibox}
\usepackage{pifont}
\usepackage{color}
\usepackage{xcolor}
\usepackage{subfiles}
\usepackage[utf8]{inputenc}
\usepackage[english]{babel}
\usepackage{xcolor}
\usepackage{ulem}
\usepackage{adjustbox}
\definecolor{blue}{RGB}{10, 10, 200}
\usepackage {fancyhdr}
\usepackage{makecell}
\usepackage[numbers,sort&compress]{natbib}
\def\BibTeX{{\rm B\kern-.05em{\sc i\kern-.025em b}\kern-.08em
    T\kern-.1667em\lower.7ex\hbox{E}\kern-.125emX}}
\input{macro.tex}

\begin{document}

\title{\sysname: Practical Cross-Chain Exchange via Trusted Hardware}
\author{
    Xiaoqing Wen, 
    Quanbi Feng, 
    Jianyu Niu,~\IEEEmembership{Member, ~IEEE,}
    Yinqian Zhang,~\IEEEmembership{Member, ~IEEE,} \\
    Chen Feng,~\IEEEmembership{Member, ~IEEE}
    \IEEEcompsocitemizethanks{
        \IEEEcompsocthanksitem Xiaoqing Wen and Chen Feng are with Blockchain@UBC and the School of Engineering, The University of British Columbia (Okanagan Campus), Kelowna, BC, Canada. Email: xqwen@student.ubc.ca and chen.feng@ubc.ca. 
        
        Quanbi Feng, Jianyu Niu, and Yinqian Zhang are with the Engineering and Research Institute of Trustworthy Autonomous Systems and the Department of Computer Science and Engineering, Southern University of Science and Technology, Shenzhen, China.
        Email: 12011501@mail.sustech.edu.cn, niujy@sustech.edu.cn and yinqianz@acm.org.
    }
}
\sloppy

\maketitle
\IEEEtitleabstractindextext{
\begin{abstract}
The proliferation of blockchain-backed cryptocurrencies has sparked the need for cross-chain exchanges of diverse digital assets. 
Unfortunately, current exchanges suffer from high on-chain verification costs, weak threat models of central trusted parties, or {synchronous}
requirements, making them impractical for currency trading applications. 
In this paper, we present \sysname, a practical cryptocurrency exchange that is trust-minimized and efficient without online-client requirements.
{\sysname leverages Trusted Execution Environments (TEEs) to shield participants from malicious behaviors, eliminating the reliance on trusted participants and making on-chain verification efficient.}
Despite the simple idea, building a practical TEE-assisted cross-chain exchange is challenging due to the security and unavailability issues of TEEs. 
\sysname tackles the unavailability problem of TEEs by implementing an efficient challenge-response mechanism executed on smart contracts. 
Furthermore, \sysname utilizes a lightweight transaction verification mechanism and adopts multiple optimizations to reduce on-chain costs. 
Comparative evaluations with XClaim, ZK-bridge, and Tesseract demonstrate that \sysname significantly reduces on-chain costs by approximately 67.87\%, 45.01\%, and 47.70\%, respectively. 
\end{abstract}

}

\IEEEdisplaynontitleabstractindextext
\IEEEpeerreviewmaketitle

\vspace{1cm}
\IEEEraisesectionheading{\section{Introduction}\label{sec:intro}}

\IEEEPARstart{I}{n} 2008, Satoshi Nakamoto invented Bitcoin~\cite{nakamoto2008bitcoin}, the first widely deployed, decentralized cryptocurrency. Subsequently, a myriad of cryptocurrencies backed by blockchain technology (including Ethereum~\cite{ethereum}, Ripple~\cite{schwartz2014ripple}, Algorand~\cite{gilad2017algorand}, etc.) have emerged for diverse applications, reshaping the digital finance landscape~\cite{fang2022cryptocurrency, liu2022common}. Remarkably, as of May 2024, nearly ten thousand cryptocurrencies exist in the market, collectively valued at approximately 2360 billion USD \cite{coinmarket}. 
This unprecedented diversity has significantly fuelled the interoperability demand between different blockchains, especially secure and efficient cross-chain cryptocurrency trading~\cite{belchior2021survey}. 

Cross-chain cryptocurrency exchange~\cite{augusto2023sok} that supports cryptocurrency trading among different blockchains typically operates as follows.
Suppose Alice would like to trade 1 Ethereum coin (ETH)~\cite{ethereum} for EOS~\cite{eos} at an exchange rate of 1 ETH = 2500 EOS, as depicted in \figref{fig:overview}. 
This trade involves two on-chain transactions: an ETH deposit from Alice to the exchange and an EOS transfer from the exchange back to Alice. 
These two transactions are inherently coupled in that they either both take place or not at all---an all-or-nothing requirement known as atomicity~\cite{herlihy2018atomic, hu2023ivyredaction}. 
This atomicity constitutes a central challenge for the design of cross-chain currency exchange because this property is proved to be impossible without a trusted third party in the asynchronous setting (\ie, no online requirements for the involved parties)~\cite{zamyatin2021sok}.

\begin{figure}[t]
    \centering
    \includegraphics[width=9cm]{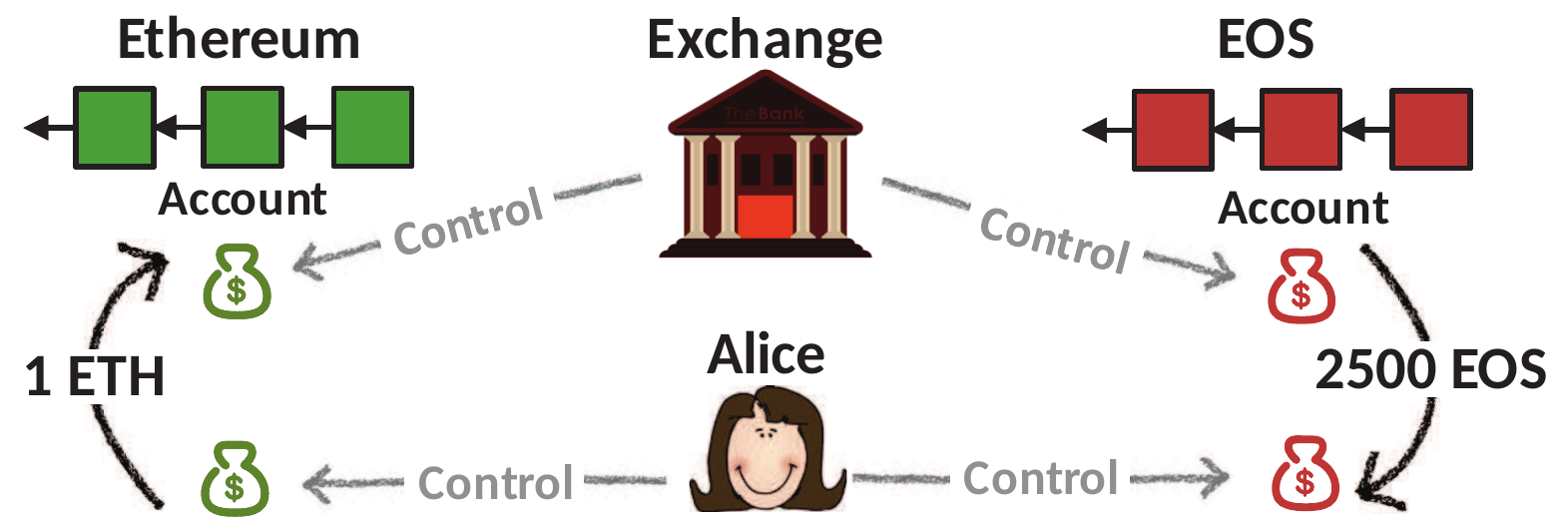}
    \caption{An example of the cross-chain exchange for Alice to exchange 1 ETH for 2500 EOS. }
    \label{fig:overview}
\end{figure}



However, current cross-chain solutions suffer from high on-chain verification costs, weak threat models of central trusted parties, or {synchronous}
requirements, making them impractical for currency trading applications. 
Specifically, 
notary-based approaches~\cite{yin2022bool, xiong2022notary} use a single or a small group of notaries (\ie, centralized parties) to handle cross-chain transactions. However, they place high trust in notaries, which makes them suffer from centralization risks~\cite{chen2023pacdam}.
Hashed Time-Lock Contract (HTLC)-based approaches~\cite{herlihy2018atomic, tesseract} leverage hash and time lock to ensure the atomicity of clients who keep continuous synchronization with blockchains (\ie, synchronization requirement).
In addition, the time lock will be designed to be relatively long (\eg, several hours~\cite{lightingnetwork}), which  
gives the initiator the priority to start or cancel the transaction, resulting in griefing attacks~\cite{thyagarajan2022universal, xu2021game}.
Relay-based approaches~\cite{zamyatin2019xclaim, xie2022zkbridge} rely on smart contracts to do costly on-chain verification of valid or fraud proofs.  
For instance, the gas consumption of verifying the zero-knowledge proof of Cosmos~\cite{cosmos} status on Ethereum is 227K gas~\cite{xie2022zkbridge}, which is about 10 times that of a simple ETH transfer on Ethereum. 
Therefore, we ponder the following question: \textit{how to design a practical exchange that efficiently proceeds cross-chain transactions without centraized parties and synchronization requirements?}
 
In this paper, we present \sysname\footnote{\sysname is a major god religion of \textit{financial gain}, \textit{commerce}, communication in Roman religion and mythology~\cite{Mercury}.}, a practical cryptocurrency exchange that is trust-minimized (\ie, making no trust assumption on exchange operators),
 and efficient (\ie, low gas fees and short confirmation latency), without online-client requirement (\ie, no network-synchronization requirements for clients).
To this end, \sysname leverages a group of operators running consensus inside Trusted Execution Environments (TEEs) 
to process the various operations of the system, eliminating the reliance on trusted operators and removing synchronization requirements for clients. 
Specifically, we use Raft protocol inspired by the previous work~\cite{engraft}, which proved that running Raft~\cite{ongaro2014search} inside TEEs can achieve Byzantine tolerance. 
Note that operators spend a client's deposits on one blockchain (\ie, the source blockchain) only after transferring corresponding currencies to the client on another blockchain (\ie, the target blockchain).  The integrity of TEE ensures that a client's deposit will not be spent until the transfer is made, thereby achieving atomicity. 

Unfortunately, directly porting cross-chain exchange processes inside the TEE cannot lead to a practical exchange due to the TEEs' unavailability and resource limitations.
First, since the number of operators in the exchange is usually small (\eg, tens of entities), there may not be enough honest operators to ensure the availability of the system. 
This is because TEEs' I/O is manipulated by their hosts (\ie, hypervisors or software systems) such that malicious operators can easily block messages from/to TEEs. 
Without availability, the system cannot ensure the atomicity of cross-chain exchange.
For instance, in the above example, after Alice deposits 1 ETH to the exchange, the host can refuse to execute the EOS transfer to Alice, \ie, being unavailable. 
As a result, Alice neither gets 2500 EOS nor withdraws 1 ETH on the source blockchain. 
Second, TEEs also face resource limitations, particularly in storage space, presenting a challenge in efficiently verifying the finalization of transactions. 
Besides,  computation and storage on the blockchain are expensive. 
For instance, storing a 256-bit integer on the Ethereum smart contract costs 8K gas (0.46 USD)~\cite{solidity}. 
Therefore, minimizing the calculation and storage operations within smart contracts is necessary.

To tackle the first issue, we propose a \textit{cost-efficient challenge-response} mechanism implemented by on-chain smart contracts, which provides the on-chain method for clients to withdraw the deposit. 
In particular, if the transfer of 2500 EOS fails, Alice can withdraw the 1ETH by initiating a challenge.
This mechanism enforces transaction atomicity even when the exchange experiences periods of unavailability. 
Second, \sysname proposes a \textit{lightweight verification} mechanism by outsourcing some workload to on-chain smart contracts.
Specifically, when registering on the smart contract, operators with TEE synchronize with the blockchain and submit the proof of the latest block to smart contact.
Thus, operators only need to synchronize the latest block rather than all historical blocks. 

To minimize the on-chain storage cost, we adopt the \textit{periodical checkpoint} for transactions.
Specifically, operators submit checkpoints regularly, updating transactions that can be challenged.
We also upload transactions and state updates to the blockchain in batch, requiring only a single multi-signature verification on-chain for multiple transactions or state updates. 
Moreover, we introduce an \textit{pledge incentive mechanism} to prevent malicious clients from orchestrating Denial-of-Service (DoS) attacks by inundating the system with extensive challenge requests.
Note that although \sysname is built atop blockchains supporting smart contracts, it can be further extended to blockchains without smart contracts by using Hash time-lock contract (HTLC). 

To validate our approach, we build a prototype of \sysname using Golang~\cite{golang}. 
We develop the on-chain smart contracts using Solidity 0.8.0~\cite{solidity} and deploy them on the Ethereum test network, Sepolia~\cite{sepolia}. 
Ethereum’s Virtual Machine (VM) provides Turing completeness~\cite{ethereum}, fulfilling the requirements for managing the accounts managed by TEE-equipped operators (\secref{sec:design}). 
We deploy \sysname on AWS and conduct extensive experiments to evaluate its performance. We also compare \sysname with three state-of-the-art counterparts, XClaim~\cite{zamyatin2019xclaim}, ZK-Bridge~\cite{xie2022zkbridge} and Tesseract~\cite{tesseract}.

\bheading{Contributions.} The contributions of this work are as follows.

\begin{packeditemize}
\item We propose \sysname, a high-performance exchange built on TEEs to enable seamless cryptocurrency trading across diverse blockchains. 
\sysname is trust-minimized and operates without online requirements for clients with the consideration of security and availability issues of TEEs. 


\item We introduce an efficient challenge-response mechanism, ensuring clients can withdraw their currencies when the system experiences unavailability. We further optimize the on-chain cost by
batching transactions and using checkpoints.

\item We discuss the extension of \sysname to blockchains without smart contracts by replacing the challenge mechanism with Hash time-lock contracts (HTLCs).

\item We evaluate \sysname and compare its performance against three counterparts. 
Our results show that \sysname significantly outperforms ZK-Bridge in terms of the transaction proof generation time (3 seconds vs. 18 seconds) and boasts a significantly lower on-chain verification and execution cost (125K gas compared to XClaim's 389K gas and Tesseract's 239K gas).

\end{packeditemize}

\section{Related Work} \label{sec:related} 
We review existing cross-chain methods for realizing cross-chain exchanges that support trading cryptocurrencies against one another. 
These methods can be divided into three categories: notary-based~\cite{yin2022bool, xiong2022notary}, HTLC-based~\cite{herlihy2018atomic, tesseract}, and relay-based approaches~\cite{zamyatin2019xclaim, xie2022zkbridge}. 
As mentioned in \secref{sec:intro}, these methods fall short in weak threat models, synchronization requirements, and high costs.
Besides, the reliance on the central party also introduces a single-point-of-failure (SPF), and some methods are tailored for specific blockchains and incompatible with other platforms. 
A comparison between \sysname and prior works is provided in Table~\ref{table:comparison}.

\begin{table}[t]
\centering
    \begin{threeparttable}
    \caption{{\textbf{Comparison of existing cryptocurrency exchange approaches.}} 
    The Comp, 
    OCR, FR, and NSPF are short for compatibility, 
    online-client requirement, fairness, and no single-point failure, respectively.
    {The compatibility denotes the ability to support heterogeneous blockchains. 
    Online-client requirement signifies that there is no synchronization requirements for clients, while fairness means that no party can benefit from the fluctuation of the exchange rate.}
    }
    \label{table:comparison}
    \begin{tabular}{@{}lcccccc@{}}
    \toprule
    Protocols &  \thead{Trust-minimized} &  CMP & \thead{OCR} & FR &  NSPF  \\
     \midrule
    HTLC \cite{herlihy2018atomic}     & \ding{51} & \ding{51} & \ding{51} & \ding{55} & \ding{51}   \\
    Tessercat \cite{tesseract}        & \ding{51} & \ding{51} & \ding{55} & \ding{55} & \ding{55} \\
    XCliam \cite{zamyatin2019xclaim}  & \ding{51} & \ding{55} & \ding{55} & \ding{51} & \ding{51} \\
    ZK-Bridge \cite{xie2022zkbridge}          & \ding{51} & \ding{55} & \ding{55} & \ding{51} & \ding{51} \\
    Bool Network \cite{yin2022bool}   & \ding{55} & \ding{51} & \ding{51} & \ding{55} & \ding{51} \\
\midrule
\textbf{\sysname}                     & \ding{51} & \ding{51} & \ding{55} & \ding{51} & \ding{51} \\
\bottomrule
\end{tabular}
\end{threeparttable}
\end{table} 

\bheading{Notary-based method.}
Notary-based method
introduces a single or a small group of notaries to proceed with cryptocurrency transfers and exchanges between different blockchains.
When a cross-chain transaction is initiated, it is firstly verified by notaries and then recorded on the involved blockchains.
Centralized notary-based exchanges~\cite{binance, coinbase} offer advantages in terms of trading fees and efficiency.
But, they suffer from significant security risks, including hacking and fraud attacks~\cite{cenex,cenex1}.
Other solutions adopt multiple nodes as a committee for processing cross-chain transactions.
Yin \etal~\cite{yin2022bool} propose Bool Network, a distributed platform that utilizes Multi-Party Secure Computing to realize distributed key management for hidden committees.
Xiong \etal~\cite{xiong2022notary} implement BFT consensus among committee members.
However, they place high trust in notaries, which makes them suffer from centralization risks~\cite{chen2023pacdam}.
Since the number of nodes in these systems is usually small (\ie, tens of nodes~\cite{ren2023interoperability}), an adversary can easily corrupt all of them.

\bheading{HTLC-based method.}
HTLC-based method
utilizes time locks and hash locks to achieve cross-chain atomic swaps.
HTLC~\cite{herlihy2018atomic} requires synchronous clients (\ie, it is interactive, necessitating all parties to be online and actively monitor all relevant blockchains during execution).
{Tesseract~\cite{tesseract} leverages TEE to generate transaction pairs of HTLC, alleviating the requirements for clients to continuously synchronize with the blockchain.}
However, in Tesseract,  a single server carries out transaction processing, which is vulnerable to a single point of failure. 
{Besides, they suffer from griefing attacks~\cite{thyagarajan2022universal, xu2021game} caused by the ﬂuctuation of the exchange rate. 
Specifically, the trader initiating the transaction can watch the exchange rate. 
If the exchange rate is favorable, the trader will finish the transaction; otherwise, the trader will opt not to disclose the secret, resulting in the abortion of the transaction.} 

\bheading{Relay-based method.} Relay-based method
adopts the relay model~\cite{deng2018research, yin2021sidechains} to deal with cross-chain exchanges, which use on-chain verification to synchronize two chains. 
XClaim~\cite{zamyatin2019xclaim} and BTCrelay~\cite{btcrelay} leverages a light client of Bitcoin on Ethereum to verify the transactions on Bitcoin with Simplified Payment Verification (SPV)~\cite{spv}.
ZK-Bridge~\cite{xie2022zkbridge} and Topos~\cite{gauthier2022topos} utilize zero-knowledge proof technology to enable the on-chain proof for the headers of Ethereum.
ZeroCross~\cite{li2022zerocross} further provides privacy protection. {CrossChannel~\cite{luo2024crosschannel} establishes a cross-chain channel that employs the chain relay to synchronize channel-related information.}
However, on-chain verification for zero-knowledge proof or SPV greatly increases the cost and time consumption for cross-chain transactions.
Due to the requirements for the verification scheme for a specific blockchain, these methods are usually tailored for a specific blockchain, making them incompatible with another blockchain. 
Cosmos~\cite{cosmos}, Polkadot~\cite{polkadot} and Agentchain~\cite{hei2022practical} introduce relay chain, a central hub to transfer and communicate cryptocurrencies between different blockchains.

\section{Problem Statement, Model and Challenges} \label{sec:model}
In this section, we first present the problem statement and the system model. Then, we introduce the associated challenges and necessary preliminaries. 
Table~\ref{table:notation} lists the frequently used notations in this paper. 

\subsection{Problem Statement}
\sysname enables a client with cryptocurrency $s$ on a source blockchain, denoted as $\mathcal{S}$, to trade for cryptocurrency $t$ on a target blockchain, denoted as $\mathcal{T}$. 
Specifically, a group of operators, forming an exchange $\mathcal{E}$, receive cryptocurrency $s$ from the client on the blockchain $\mathcal{S}$ and transfer cryptocurrency $t$ to the client on the blockchain $\mathcal{T}$. 
The account of the exchange $\mathcal{E}$ on the blockchain $\mathcal{S}$ (resp., $\mathcal{T}$) is denoted by \textsf{vault$_S$} (resp., \textsf{vault$_T$}). 
To achieve cross-chain exchange, we need two transactions: 1) the client deposits currency $s$ from its account to \textsf{vault$_S$} on the blockchain $\mathcal{S}$; and 2) the exchange $\mathcal{E}$ transfers currency $t$ from its account \textsf{vault$_T$} to the client's account on the blockchain $\mathcal{T}$.
We use $tx_S$ and $tx_T$ to denote the above two transactions, respectively. 
Here, clients and operators use their public/private key pairs as their identifiers on both the blockchain $\mathcal{S}$ and $\mathcal{T}$, and their transactions can be verified by their signatures.  
We assume the client can synchronize with the blockchain $\mathcal{S}$ and $\mathcal{T}$ to initiate on-chain transactions. 

\begin{table}[t]
    \footnotesize
    \centering
    \setlength{\abovecaptionskip}{0cm}
    \caption{\textbf{Summary of notations.}} \label{table:notation}
    \begin{tabular}{@{}m{0.9cm}l|m{0.9cm}l@{}}
        \toprule[1pt]
        Term    & Description  &  Term    & Description   \\
        \midrule
        $\mathcal{S}$         & Source blockchain           & $p_i$ & Operator in \sysname\\
        $\mathcal{T}$         & Target blockchain           & $n$ & Number of operators       \\
        \textsf{vault$_S$}   & Vault on source chain  & $\eta_i$    & Enclave for operator $p_i$       \\ 
        \textsf{vault$_T$}   & Vault on target chain  & $\tau_c$       & Timer for a challenge\\
        $tx_S$      & Trans on source chain  & $\tau_w$ & Timer for response      \\
        $tx_T$      & Trans on target chain  & $(pk_i, sk_i)$ & Key pair for enclave $i$\\
        \bottomrule[1pt]
    \end{tabular}
    \vspace{-0.4cm}
\end{table}



\subsection{System Model}\label{subsec:model}
We consider an exchange of $n = 2f+1$ operators, denoted by the set $\mathcal{P} = \{p_1, p_2, ..., p_n\}$. 
Each operator $p_i$ is equipped with a TEE machine. 
We assume all operators can arbitrarily deviate from the protocol, while the software inside TEEs is off-limit (introduced shortly). 
The underlying reason for making this harsh assumption is that in a real-world exchange, the number of operators is usually small (\ie, tens of operators~\cite{ren2023interoperability}), and all of them may be simultaneously corrupted. 
{What is more, we consider the worst case, in which all Byzantine operators can be controlled by a single adversary $\mathcal{A}$.}
There is a public/private key pair of the operator $p_i$, denoted by $(pk_i, sk_i)$, in which the private key is only stored and used inside  TEE.



\bheading{Threat model of TEE.}
We assume the TEE to protect the program inside TEE in line with other TEE-assisted designs~\cite {zhang2016town, cheng2019ekiden, engraft, tesseract}.
We assume that TEE provides integrity and confidentiality guarantees, and secure remote attestation, which can provide unforgeable cryptographic proof that a specific program is running inside TEE~\cite{sgxremoteattest, sevremoteattest}.
Also, an operator cannot know the private key $sk_i$ inside TEE.
We assume operators have full control over those machines, including root access and control over the network.
The operators can arbitrarily launch, suspend, resume, terminate, and crash TEEs at any time. 
Besides, the TEE operators can delay, replay, drop, and inspect the messages sent to and from TEEs, \ie, manipulating input/output messages of TEEs. 
In other words, the TEE cannot guarantee availability due to these manipulations, which is the main issue addressed in this work. 
The rollback attacks are orthogonal to this work and can be addressed by work~\cite{matetic2017rote, brandenburger2017rollback, narrator, narrator-pro}.



\bheading{Blockchain model.} 
We establish a cross-chain exchange between two blockchains with different consensus mechanisms and models.
As the execution of transactions relies on both chains, if either chain is compromised, the system cannot operate normally.
Thus, we follow the assumption in~\cite{zamyatin2019xclaim} that the blockchain $\mathcal{S}$ and blockchain $\mathcal{T}$ are both safe and live.
Here, the safety property means that all the honest nodes have the same view of the blockchain, while the liveness property refers to the fact that some event must eventually happen in a distributed system~\cite{kindler1994safety}.  

We assume that both blockchains provide finality, \ie, once transactions or blocks are finalized, they will remain finalized forever~\cite{zhang2020analysis}.
We assume that both blockchains support smart contracts~\cite{zou2019smart}.
Smart contracts can access the current time using the method $block.timestamp$ and the hash of the recent 265 blocks~\cite{solidity} using the method $block.bh(i)$, where $i$ is the number of the block.
Smart contracts can also
{support} cryptographic schemes, including hash computing and ECDSA encryption algorithm. 



\subsection{System Goals}
we derive the following desirable properties for \sysname under the above threat and blockchain models.
\begin{packeditemize}
\item \textbf{Atomicity.} 
An atomic cross-chain settlement can guarantee that both $tx_S$ and $tx_T$ will be confirmed on their chains, or neither $tx_S$ nor $tx_T$ will be confirmed on their chains.

\item\textbf{Trust-Minimized Participants.}
No trusted third parties are involved to ensure the atomicity of cross-chain exchange.

\item\textbf{No on-lineclient requirement.}
If the clients do not synchronize with the blockchain after initializing the cross-chain requests, they will not suffer any loss (\ie, the atomicity can be guaranteed).
\end{packeditemize}









\subsection{Challenges}
There are several challenges in designing \sysname, which are discussed below.

\noindent \textbf{Atomicity protection against malicious operators.} 
The unavailability of TEEs poses a significant threat to the overall atomicity of cross-chain exchanges. 
For example, when all TEEs are unavailable, a client, after depositing its assets on the blockchain $\mathcal{S}$, will never receive the corresponding tokens from the blockchain $\mathcal{T}$.
To address this critical issue, we propose a practical challenge-response mechanism implemented through on-chain smart contracts. 
If the client's deposit is not processed for a long time, the client can raise a challenge on the smart contract to force the operator to execute the transaction and reply to the challenge. 
Otherwise, the deposit will be returned to the client.

\noindent \textbf{Transaction verification under limited resources in TEEs.}
TEEs face resource limitations, particularly in storage space, presenting a challenge in efficiently verifying the finalization of a given transaction.
Inspired by the light-client design, which stores only block headers to verify transactions, \sysname proposes a similar solution that relies only on block headers.
Unlike traditional light-client design, \sysname doesn't require all the historical block headers and it outsources some workload to on-chain smart contracts.
Specifically, in the initiation phase, TEEs engage in a process of getting the latest finalized block headers and committees.
Notably, successful registration on the chain is exclusively granted to TEEs that proficiently synchronize historical block headers.
Therefore, adopting simple on-chain verification can reduce the synchronization cost for historical blocks in TEE.


\vspace{2mm} \noindent \textbf{Minimizing expensive on-chain costs.}
In \sysname, an on-chain vault is instrumental for interactions with operators and users, which introduces computational and storage costs that directly impact the expenses of cross-chain exchanges. Recognizing the necessity to minimize these costs, our design focuses on the efficient use of on-chain resources.
To reduce the computational cost of on-chain verification for TEEs, the TEE uploads a transaction bundle to the chain, allowing for the verification of multiple transactions with a single TEE signature.
To minimize on-chain storage costs, we establish checkpoints for challengeable transactions, enabling the removal of transactions preceding the checkpoint.

\subsection{Building Blocks}
\vspace{2mm} \noindent \textbf{Trusted Execution Environment.}
An operator can instruct his TEE to create new enclaves, \ie, new execution environments running a specified program.
To develop TEE-agnostic protocols, we adopt the methodology outlined in~\cite{pass2017formal}, 
which models the functionality of attestable TEEs. 
Each TEE is initialized with a key pair $(pk, sk) \leftarrow \Sigma.keyGen(1^{\lambda})$ generated by its manufacturer, where $\Sigma$ is the signature algorithm. 
Here, $sk$ represents the TEE's internal master secret key, while $pk$ is the corresponding public key.
The ideal functionality of a TEE offers the following Application Interfaces (APIs) for a program code $prog$: 
\begin{packeditemize}

    \item $eid \leftarrow install(prog)$: Install the program $prog$ as a new enclave within the TEE, assigning it a unique ID, $eid$.

    \item $(outp, \sigma) \leftarrow resume(eid, inp)$: Resume the enclave with ID $eid$ to execute program $prog$ using the input $inp$. $\sigma$ is the TEE's endorsement, confirming that $outp$ is the valid output.
\end{packeditemize}

Besides, we include an attestation API for TEE to generate an attestation quote $\rho \leftarrow attest(eid, prog)$ proving that the program $prog$ has been installed in the enclave $eid$.
And $\rho$ can be verified through method $verifyquote(\rho)$~\cite{onchainattest}.

\vspace{2mm} \noindent \textbf{Cryptographic Primitives.}
Our protocol utilizes a multi-signature scheme ($\textsc{GenSig}, \textsc{Sign}, \textsc{Verify}$), and a secure hash function $\textsc{Hash}(\cdot)$
All message sent within our protocol are signed by the sending party. 
A message signed by party $p$ is denoted as  $\textsc{Sign}(mr;sk_p) \rightarrow \sigma$.
Multi-signature scheme allows multiple parties to collaboratively generate a digital signature for a message by two main functions: \textsc{Sign} and \textsc{Verify}. 
Specifically, given a set of users $G$ and the message $m$, the function \textsc{Sign}($mr, \{sk\}_{g \in G}$) $\rightarrow \sigma$ takes the private key $sk$ and the hash of the message $mr$ as input and return the signatures $\sigma$. 
The function \textsc{Verify}($\{pk\}_{g \in G}, mr, \sigma$) $\rightarrow 0/1$ takes the public key $pk$, $mr$ and $\sigma$ and returns the result of whether the signature was generated by $G$.

\vspace{2mm} \noindent \textbf{Automated Market Maker (AMM).}
AMM~\cite{mohan2022automated} algorithm automatically provides liquidity to a decentralized exchange by continuously adjusting the exchange rate of the traded currencies based on some predefined mathematical formulas.
Specifically, Constant Function Market Maker (CFMM)~\cite{adams2021uniswap}
can keep the product of the quantities of currencies as constant.
For instance, in a system which holds $\mathrm{X}$ (resp., $\mathrm{Y}$) amount of currency $\mathrm{A}$ (resp., $\mathrm{B}$), if a client wishes to exchange $x$ amount of currency $\mathrm{A}$ for currency $\mathrm{B}$, the quantity of currency $\mathrm{B}$ available for exchange satisfies $(\mathrm{X}+x)*(\mathrm{Y}-y) = \mathrm{X}*\mathrm{Y} = k$, where $k$ is a constant.
The function $y \leftarrow \textsc{AMM}.exchange(x)$ calculates the exchange rate from $\mathrm{A}$ to $\mathrm{B}$.
The currencies in the liquidity pool are supplied by Liquidity Providers (LPs)~\cite{aigner2021uniswap}.
LPs contribute to the liquidity pool by depositing both currencies $\mathrm{A}$ and $\mathrm{B}$. 


\section{\sysname Design} \label{sec:design}

\subsection{Overview}\label{overview}
\begin{figure}[t]
    \centering
    \includegraphics[width=9.2cm]{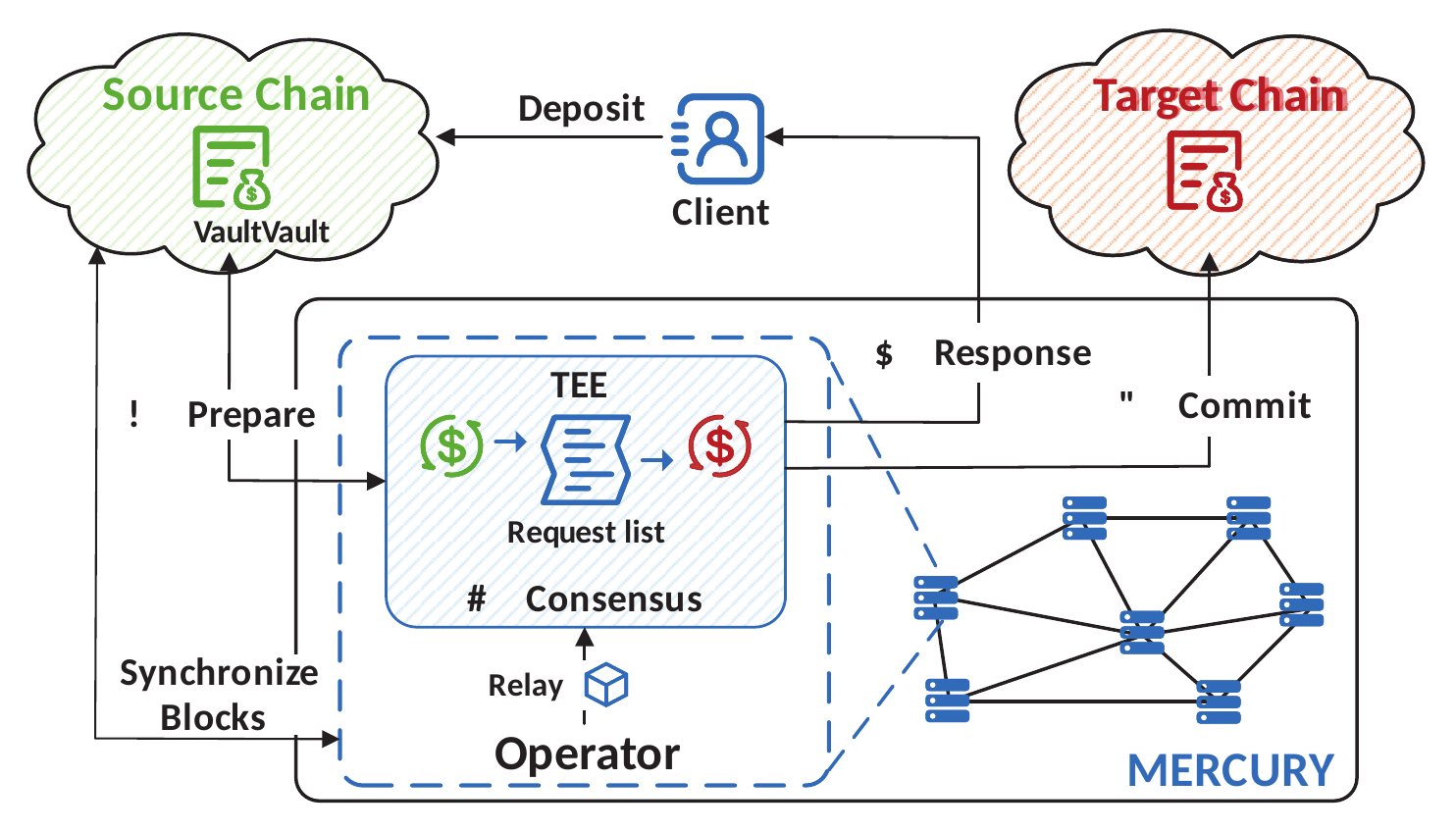}
    \caption{An architectural overview of \sysname. A client transfers currencies from the chain $\mathcal{S}$ to the chain $\mathcal{T}$. 
    }
    \label{fig:TExchange}
\end{figure}

\figref{fig:TExchange} illustrates an architectural overview of \sysname, which enables secure currency exchanges between any two chains: the blockchain $\mathcal{S}$ and the blockchain $\mathcal{T}$. 
The vaults act as the on-chain interactions for clients, while operators are the entities that handle cross-chain exchanges.

\sysname includes five key components: system initialization (\secref{sec:init}), normal-case operation of the cross-chain exchange (\secref{sec:normalcase}), challenge-response mechanism (\secref{sec:failures}), {transaction verification inside TEE} (\secref{sec:synchronization}), and rewards (\secref{sec:reward}).
Specifically, system initialization involves the registration of operators on the respective blockchains; 
normal-case operations are responsible for transaction processing;
challenge-response mechanism ensures clients redeem their deposits on the chain;
transaction verification allows enclaves to synchronize with blockchain $\mathcal{S}$ and $\mathcal{T}$;
and incentive mechanisms to incentivize liquidity provision for the vaults and to sustain services, equitably distributing transaction fees among LPs and operators. 

\begin{algorithm}[t]
\caption{The Partial Functions Executed by \textsf{vault}}\label{alg:vault}
\begin {algorithmic}[1]

\State{\textbf{Function}~\textbf{\textsc{Init}($committee$)}}
\Comment{Initialize Requests}
    \State{~~~~$Dep \leftarrow \emptyset$}\Comment{Set deposits list to empty}
    
\State               

    

\State{\textbf{Function~\textsc{Deposit}($s_{addr}, v$)} } 
    \State{~~~~\textbf{Require} $value > 0$}
    \State{~~~~$id \leftarrow $ \textbf{\textsc{Hash}($s_{addr}, v, block.timestamp$)}}
    \State{~~~~$Dep_{id}.sender \leftarrow s_{addr}$}
    \State{~~~~$Dep_{id}.value \leftarrow v$}
    \State{~~~~$Dep_{id}.underChal \leftarrow false$}
    \State{~~~~Trigger Deposit \textbf{event}}

 \State
 \State{\textbf{Function~\textsc{Transfer}($tranSet, signs$)}} 
    \State{~~~~\textbf{Require} \textbf{\textsc{verifyMutisig}($signs$)}}
    \State{~~~~\textbf{For} each $ tx \in tranSet$}
    \State{~~~~~~~~$tx.recever.transfer(v)$}
    \State{~~~~\textbf{End For}}

\State
\State{\textbf{Function~\textsc{StartChallenge}($id, value, pledge$)}} 
    \State{~~~~\textbf{Require} $pledge \geq Pledge_C$}    
    \Statex{~~~~~~~~~~~~~~~$\wedge value = Dep_{id}.value$}
    \Statex{~~~~~~~~~~~~~~~$\wedge Dep_{id}.underChal = false$}
    \State{~~~~$Dep_{id}.underChal \leftarrow true$}
    \State{~~~~$Dep_{id}.startChal \leftarrow block.timestamp$}
    \State{~~~~Trigger Challenge \textbf{event}}

\State
\State{\textbf{Function~\textsc{ResloveChallenge}($id$)}} 
    \State{~~~~\textbf{Require} $Dep_{id}.underChal = true$}
    \Statex{~~~~~~~~~~~~~~~$\wedge block.time - Dep_{id}.startChal > \tau_w$}
    \State{~~~~\textbf{\textsc{Refund}($id$)}}
    \State{~~~~\textbf{\textsc{Delete}($Dep_{id}$)} }

    
\State
\State{\textbf{Function~\textsc{UpdateCheckpoint}($idSet,signs$)}} 
    \State{~~~~\textbf{Require} \textbf{\textsc{verifyMutisig}($signs$)}}
    \State{~~~~\textbf{For} each $id \in idSet$}
    \State{~~~~~~~~\textbf{\textsc{Delete}($Dep_{id}$)} }
    \State{~~~~\textbf{End For}}
    
\end{algorithmic}
\end{algorithm}

\subsection{System Initialization}\label{sec:init}
During the initialization,
$vaults$ are deployed on different blockchains by the service providers.
Algorithm~\ref{alg:vault} gives the main functions executed in \textsf{vault}.
Then, the operators join the system by registering on $vaults$.
Thus, $vaults$ are jointly controlled by the operators, whose public keys are recorded on the $vaults$.
The registration process involves the attestation and block synchronization verification for the enclave, and recording the enclave's public key on the $vaults$.

There are $n=2f+1$ operators that register in the $vaults$ during the initialization.
Here, we present an exemplary representation of an operator $p_i$ registering on \textsf{vault$_S$} for clarity.
The operator $p_i$ equipped with TEE instructs his TEE to install the program $prog$ for \sysname and creates a new enclave $\eta_i$.
The registration process has three steps:

\vspace{1mm} \noindent \blackding{1} 
The enclave $\eta_i$ in operator $p_i$ generates a unique pair of private and public keys $\left<sk_i, pk_i\right>$ for the \textsf{vault$_S$} on the blockchain $\mathcal{S}$. 
The secret key $sk_i$ is securely stored within $\eta_i$, accessible only by the program running in $\eta_i$, while the public key $pk_i$ is returned as output to the operator $p_i$. 
Subsequently, $\eta_i$ generates an attestation $\rho_i$ claiming that the program $prog$ runs in  $\eta_i$ and controls the secret key $sk_i$. 

\vspace{1mm} \noindent \blackding{2} 
The operator $p_i$ transmits the header of the latest $k$ finalized block to the enclave $\eta_i$, where $k$ is set as the upper bound the enclave may lag. 
$\eta_i$ verifies the block header's consistency (refer to ~\secref{sec:synchronization}), generating a proof $\rho_i^{bc}$ (\ie, the header of the latest block signed by the $\eta_i$). 

\vspace{1mm} \noindent \blackding{3} 
Then, the operator registers $\eta_i$ by invoking \textsc{Register}$(\eta_i, \rho_i, \rho_i^{bc}, pk_i)$ on the \textsf{vault$_S$}. 
The \textsf{vault$_S$} verifies that $\rho_i$ is a valid attestation~\cite{pass2017formal, onchainattestation} and $\rho_i^{bc}$ refers to one of the latest $k$ blocks on the blockchain $\mathcal{S}$.
Upon successful verification, \textsf{vault$_S$} adds the operator $p_i$ with $pk_i$ to the operators list. 

The registration process guarantees the execution of the $prog$ for all registered enclaves and the confidentiality of the secret key $sk_i$. 
Consequently, there is no need to re-attest enclaves in subsequent protocol steps.
Clients can access the \textsf{vault$_S$} by its smart contract address on the blockchain $\mathcal{S}$.

\subsection{Cross-Chain Exchange}\label{sec:normalcase}


To safeguard the liveness of \sysname, no more than $f$ TEEs operated by a group of $n$ operators may crash.
To ensure that operators can generate enough votes and that merely $f$ crash operators cannot collect valid numbers of votes, we design the system such that $n = 2f + 1$, and a minimum of $f + 1$ votes is required for passing the verification on smart contract.
As mentioned in~\secref{sec:intro}, we implement RAFT protocol to the group of TEE-equipped operators.  
In the normal-case operation of a cross-chain exchange, where a client seeks to trade its currencies between the blockchain $\mathcal{S}$ and the blockchain $\mathcal{T}$, the process encompasses five phases 
as shown in~\figref{fig:TExchange} and explained below.

\vspace{1mm} \noindent \blackding{1}
The client initiates an exchange request by depositing its currency $s$ to \textsf{vault$_S$} on blockchain $\mathcal{S}$. 
Specifically, it creates a $deposit$ transaction as a tuple of $\left<s_{addr}, v\right>_\sigma$(cf. Algorithm~\ref{alg:vault} line 4), where 
$s_{addr}$ is the client's address on the blockchain $\mathcal{S}$, and \textsf{vault} is the transferred value of currency $s$ (greater than zero).
\textsf{vault$_S$} calculates a unique $id$, which is the hash value of the $deposit$ transaction, records  $\left<id, s_{addr}, v\right>$, and locks the \textsf{vault} amount of currency $s$.
Note that each $deposit$ corresponds to a unique on-chain identifier $id$, which prevents replay attacks (\ie, clients cannot initiate more than one $request$ using the same $deposit$).

Then, the client sends a $request$ as a tuple $\left<\textsc{request}, id, s_{addr}, id_T, r_{addr} \right>_\sigma$ to all the operators, where $id_T$ is the identifier of the blockchain $\mathcal{T}$, and $r_{addr}$ is the receiving address of the client on $\mathcal{T}$. 
Note that the message $request$ contains the information about the blockchain $\mathcal{T}$, which is not included in the $deposit$ transaction to reduce on-chain transaction size and associated fees. 


\vspace{1mm} \noindent \blackding{2} 
Upon receiving the client's $request$, the enclave $\eta_i$ verifies whether the $deposit$ transaction is finalized on the blockchain $\mathcal{S}$ by synchronizing with $\mathcal{S}$ (details in \secref{sec:synchronization}). 
Upon successful verification, operators check whether  $s_{addr}$ and \textsf{vault} of the $deposit$ match the information in $request$.
Then, they compute the amount of currency $t$ by $\textsc{AMM}.exchange(v)$, where \textsf{vault} (resp., \textsf{vault$_T$}) is the amount of currency $t$ (resp., $s$).
Then, the enclave $\eta_i$ generates a $transfer$ transaction for the blockchain $\mathcal{T}$, denoted as $\left<id_T, v_T, r_{addr} \right>$, adds it to the memory pool, and awaits consensus. 


\vspace{1mm} \noindent \blackding{3} 
The enclaves in operators run a RAFT protocol to achieve agreement on the transaction $\left< \textsc{transfer}, id, s_{addr}, v, id_T, r_{addr} \right>$. 
Note that transactions, aiming at the same blockchain $\mathcal{T}$, are batched and signed by operators and subsequently submitted to the \textsf{vault$_T$}.
Thus, \textsf{vault$_T$} can process several transactions through one-time signature verification, optimizing on-chain computing cost.
During the consensus, $\eta_i$ signs batches with its private key for the blockchain $\mathcal{T}$, and broadcasts the signature $sig_i$ to other TEEs. 
This private key is generated during the initialization and the public key is recorded in the \textsf{vault$_T$}.

\vspace{1mm} \noindent \blackding{4} 
Upon collecting more than $f$ signatures, $\eta_i$ generates a $m$-$n$ multi-signature $\sigma$ for the batch of transactions, where $m=f+1$.
Then, $\eta_i$ submits these transactions together with the multi-signature $\left<txs \right>_{\sigma}$ to the blockchain $\mathcal{T}$.
Miners on the blockchain $\mathcal{T}$ verify the multi-signature and execute the transactions. 
For each transaction $transfer$, the designated receiver account, as specified by the client, receives the transfer from \textsf{vault$_T$}.

\vspace{1mm} \noindent \blackding{5} 
After $transfer$ is finalized on the blockchain $\mathcal{T}$, operators return the execution result to the respective clients associated with these transactions.

\begin{figure}[t]
    \centering
    \includegraphics[width=8.5cm]{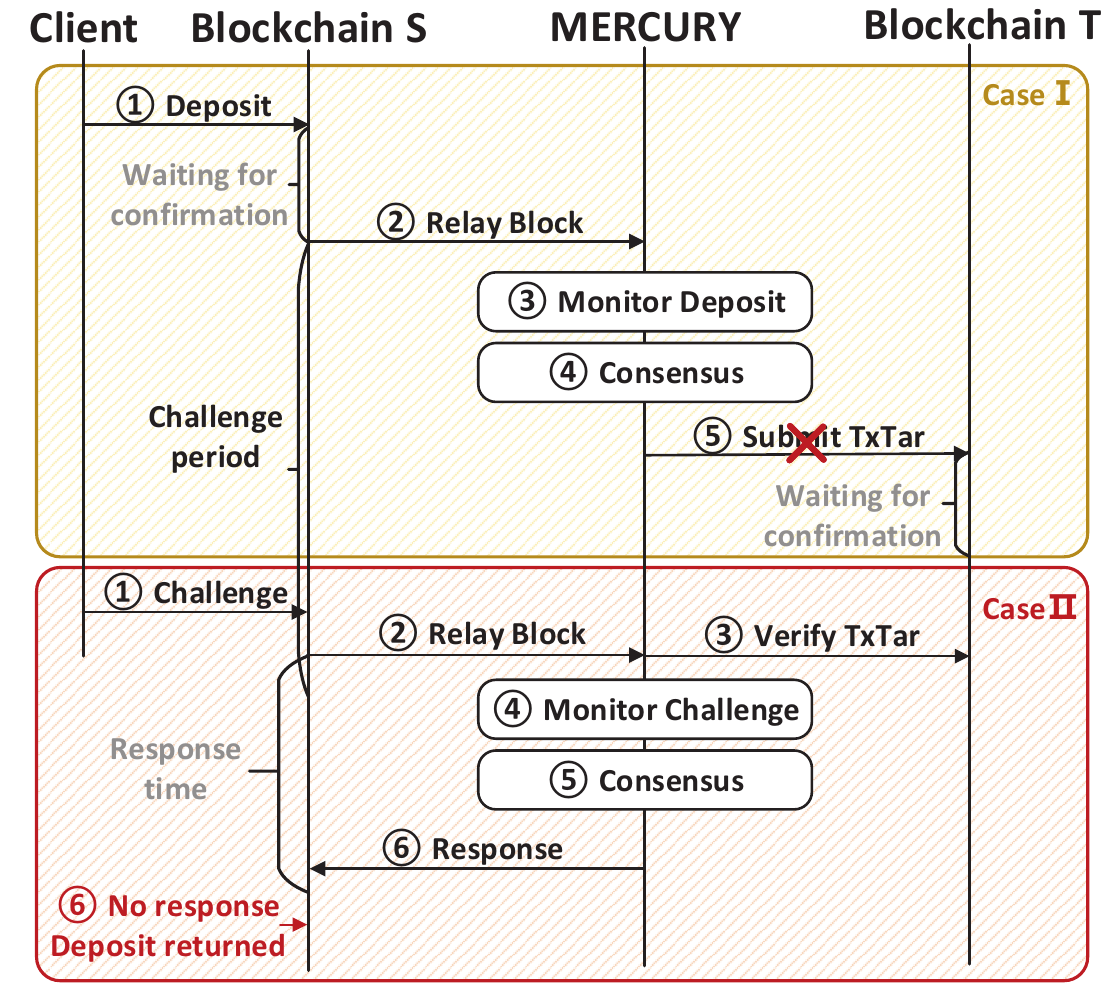}
    \caption{The design of the challenge-response protocol. Case \uppercase\expandafter{\romannumeral1} represents the normal-case operation, whereas case \uppercase\expandafter{\romannumeral2} represents the case of failures. 
    The red cross represents the process that cannot be completed when operators fail. 
    }
    \label{fig:transaction}
\end{figure}

\subsection{Safe Redemption} \label{sec:failures}
In this section, we propose a practical challenge-response mechanism deployed in \textsf{vault$_S$} to address the safe redemption issue.
As mentioned in~\secref{sec:model}, TEEs controlled by Byzantine operators cannot ensure availability.
Thus, there exists a potential scenario wherein transaction $tx_S$ on the blockchain $\mathcal{S}$ is finalized, while transaction $tx_T$ on the blockchain $\mathcal{T}$ remains unsubmitted due to the unavailability of the TEE, which breaks the atomicity of cross-chain exchange. 
The challenge-response mechanism provides an on-chain method in this scenario, which allows clients to withdraw their deposit by initiating a challenge on \textsf{vault$_S$}.


To implement the challenge-response mechanism, we made some adjustments to the normal case operation.
Upon the client's submission of a $deposit$, the currency is not directly transferred to the \textsf{vault$_T$}; instead, it remains pending. 
Thus, the client can withdraw the $deposit$ if the transaction on the blockchain $\mathcal{T}$ fails.
Only after the operators process the transaction and receive the confirmation on the blockchain $\mathcal{T}$ can the $deposit$ be locked (\ie, can not be withdrawn).

\figref{fig:transaction} illustrates the challenge process.
After creating the $deposit$ transaction, a client sets a timer of $\tau_c$. 
We can set $\tau_c$ to $2\delta_{S} + 2\delta_{E} + \delta_{T}$, where $\delta_{S}$ (resp., $\delta_{T}$) is the transaction confirmation time on the blockchain $\mathcal{S}$ (resp., blockchain $\mathcal{T}$), and $\delta_{E}$ is the time for the operators to process the exchange request.
If transfer $tx_T$ is finalized on the blockchain $\mathcal{T}$, operators will send $\left< \textsc{confirm}, id\right>_{\sigma}$ to \textsf{vault$_S$}.
Then, \textsf{vault$_S$} set the state of $tx_S$ to be confirmed and can no longer be withdrawn.
Else if no $tx_T$ on the blockchain $\mathcal{T}$ is finalized before the timer is triggered, the client invokes \textsc{StartChallenge()} (cf. Algorithm~\ref{alg:vault} line 18) on \textsf{vault$_S$}.

Once the client initiates a challenge on \textsf{vault$_S$}, a timer $\tau_w$ waiting for a response is set up.
To give the operators enough response time, $\tau_w \gg 2\delta_S + \delta_{E}$.
Note that $\tau_w$, as the waiting time for clients to refund in case of failure, will not affect the transaction latency under normal cases.
Thus, in practice, $\tau_w$ is set to be sufficiently long, providing the operators with ample time to resubmit transactions and respond.
When a challenge against $deposit_i$ is triggered in \textsf{vault$_S$}, the response falls into the following two cases.

\begin{packeditemize}
    \item \textbf{Operators are available.}
    When the operators witness the challenge from the blockchain $\mathcal{S}$, they verify whether $transfer_i$ (\ie, the corresponding transaction for the $deposit_i$ in the blockchain $\mathcal{S}$) on the blockchain $\mathcal{T}$ has been finalized.
    If $transfer_i$ has not been carried out, the operators can resubmit it on the blockchain $\mathcal{T}$. 
    Once the transaction is finalized, the operators return the result to \textsf{vault$_S$}, thereby successfully concluding the challenge. 
    The resolution of challenges similarly requires consensus, with operators reaching consensus on $\left<\textsc{challenge}, id \right >$ and collecting more than $f$ signatures.
    Then, the operators hand up the $deposit_i$ and the multi-signature of operators.
    $valut_S$ verifies the multi-signature and deletes $deposit_i$ if validated.

    \item \textbf{Operators are unavailable.} 
    The operators are unable to respond to the challenge on the blockchain $\mathcal{S}$.
    Once $T_w$ has elapsed, the client can claim a refund on \textsf{vault$_S$} by invoking \textsc{ResolveChallenge}$(id)$ (cf. Algorithm~\ref{alg:vault} line 24).
    Then, the \textsf{vault$_S$} refunds the $deposit_i$ to the client if the waiting time $\tau_w$ has elapsed.
\end{packeditemize}

To counteract potential malicious challenges, the client is required to provide a pledge while starting the challenge.
If the challenge is successful, the client receives a refund along with the pledge; otherwise, the pledge is forfeited.
Besides, to reduce on-chain storage, our approach retains only outstanding $deposits$ on \textsf{vault$_S$}, ensuring that completed ones are removed. 
After $transfer_i$ has been finalized on the blockchain $\mathcal{T}$, operators delete $deposit_i$ from $Dep$ on the blockchain $\mathcal{S}$, and the associated $request$ cannot be challenged anymore.
However, this approach increases the interactions between operators and the blockchain $\mathcal{S}$.
Thus, we set up checkpoints to batch-confirm transactions. 
Operators upload checkpoints after acquiring a fixed number of completed cross-chain exchanges or after a designated period.
Specifically, operators upload a checkpoint through \textsc{UpdateCheckpoint}$(idSet,signs)$ (cf. Algorithm~\ref{alg:vault} line 29), where $idSet$ is the list of cross-chain deposits to be confirmed and $signs$ is the $f+1$ multi-signature of operators.

\subsection{{Transaction Verification inside TEE}}\label{sec:synchronization}
To verify transactions on the blockchain, operators need to verify transactions on the blockchain $\mathcal{S}$ and blockchain $\mathcal{T}$ inside TEE.
However, given the potential for malicious operators, there exists a risk for enclaves to receive forged blocks.
Besides, the current storage requirement for a full node of Ethereum exceeds 600GB~\cite{ethwiki}, which is impractical for enclaves.
Inspired by the light-client design~\cite{ethlightclient}, we design a verification strategy in enclave $\eta_i$ that checks the validity of received blocks without withholding the whole blockchain.
Ethereum simplifies the verification for the light clients by introducing a sync committee consisting of 512 randomly selected validators. 
The members of the sync committee are expected to sign every block header at each slot (\ie, the time interval for block proposal in Ethereum), among which there is a different node assigned to propose a new block.
Besides, the current and the next syn committee can also be observed within the finalized blocks.
Thus, we adapt the design of the sync committee.


During the initialization phase, $\eta_i$ first synchronizes the existing $k$ finalized blockchain headers BHF and related block data.
$b_h$ is the latest finalized block with height $h$.
$\eta_i$ verifies that ($i$) BHF itself is consistent, ($ii$) the latest finalized block ($b_h$) has enough signatures of the sync committee.
Then, $\eta_i$ generates a proof $\rho_i^{bc} \leftarrow \left< b_h.bh \right>_{\sigma}$, which is the header of $b_h$ signed by enclaves, and hand out to the vault.
Vault verifies $\rho_i^{bc}$ and decides whether enclave $\eta_i$ has gotten the latest $k$ finalized block $b_h$ with blockchain, where $k$ is set as the upper bound to the time an enclave may lag behind.
If passed, the enclave $\eta_i$ is considered to be initially synchronized successfully.
{Then, enclave $\eta_i$ can discard block data 
and retain current and next sync committee members, along with BHF storing finalized block headers.}



After the initialization, $\eta_i$ continuously receives finalized blocks.
Once a block is finalized~\cite{buterin2020combining} and verified by $\eta_i$, its header, along with all its ancestors, is placed in BHF. 
To verify that a block $B$ is finalized, $\eta_i$ needs to verify the validity of the sync committee and their signatures.
By comparing the sync committee with the expected one conveyed in the historical finalized block, $\eta_i$ can verify the validity of the committee as well as these signatures.
Therefore, enclave $\eta_i$ only needs to store historical block headers since initialization.



\subsection{Vault Liquidity and Rewards}\label{sec:reward}
As mentioned in~\secref{sec:init}, we deploy $vaults$ jointly controlled by the operators, allowing any parties to provide $vaults$ with liquidity.
To motivate operators to diligently handle transactions and LPs to provide liquidity, we distribute transaction fees among LPs and operators. 
The idea behind this design is that incentive mechanisms such as transaction fees and block rewards play an important role in securing blockchain systems~\cite{eyal2018majority, Niu2019, crystal}. 

Suppose the transaction fee for one transaction is $F$.
The LP provides an amount $L$ of liquidity to \textsf{vault}, the total liquidity in \textsf{vault} is $\mathcal{T}$. 
The reward for the LP is then proportional to the fraction of total liquidity it provided, expressed as: $F_L = (L / T)Fr$, where $r$ is a pre-determined parameter representing the proportion of rewards in transaction fees for LPs.
For each completed transaction, operators that participate in the final $m-n$ votes of the consensus process receive rewards.
For a transaction with $m$ signatures, each TEE participating in its consensus can obtain $F_T = (1 / m)F(1-r)$.

Furthermore, to incentivize continued service provision, each operator is required to deposit upon registration. 
If \sysname is deactivated, settlement is executed on the blockchain, and any losses incurred by \sysname due to TEE unavailability will be borne by the TEE providers. 


\section{Security Analysis} \label{sec:analysis}
In this section, we conduct a security analysis of \sysname.
As described in~\secref{sec:model}, \sysname is \textit{trust-minimized} since there is no trust assumption of operators in the exchange.
Then, we give a proof of Atomicity and 
No online-client requirement.
Recall that $tx_S$ is a $deposit$ transaction on the blockchain $\mathcal{S}$ and $tx_T$ is a $transfer$ transaction. 
We establish several technical lemmas. 

\begin{lemma}\label{lemmap3g}
If $tx_S$ is confirmed on the blockchain $\mathcal{S}$, $tx_T$ must be confirmed on the blockchain $\mathcal{T}$.
\end{lemma}
\begin{proof}
Proof by contradiction. Suppose that $tx_T$ is not confirmed on the blockchain $\mathcal{T}$.
Thus, the request that confirm the associated $tx_S$ in \textsf{vault$_S$} is invalid.
As mentioned in~\secref{sec:model}, programs inside TEE are protected and can not be compromised, \ie, this invalid request can not be signed by the secret key inside the operators' TEE.
Then, the invalid request without $f+1$ signatures from operators' TEE will not be accepted by \textsf{vault$_S$} and no one can confirm $tx_S$ in \textsf{vault$_S$}.
Therefore, after a period of $\tau_c$, the client can issue a challenge and cancel the $tx_S$.
This contradicts the fact that $tx_S$ is confirmed on the blockchain $\mathcal{S}$.
\end{proof}

\begin{lemma}\label{lemmap4g}
If $tx_T$ is confirmed on the blockchain $\mathcal{T}$, $tx_S$ must be confirmed on the blockchain $\mathcal{S}$.
\end{lemma}
\begin{proof}
Proof by contradiction. Suppose that $tx_S$ is not confirmed on the blockchain $\mathcal{S}$. Then, we have two cases as follows.

\begin{packeditemize}
    \item \textbf{Case 1.} $tx_S$ is not even submitted to the blockchain $\mathcal{S}$.
    Thus, the associated transaction $tx_T$ in \textsf{vault$_T$} is invalid.
    In this case, $tx_T$ can not be signed by $f+1$ operator's TEE and can not be confirmed on the blockchain $\mathcal{T}$.
    

    \item \textbf{Case 2.} $tx_S$ is pending on the blockchain $\mathcal{S}$ (but not confirmed).
    According to~\secref{sec:failures}, operators will generate the associated $tx_T$ and submit it to the \textsf{vault$_T$}.
    If transaction $tx_T$ is not confirmed on blockchain $\mathcal{T}$, the assertion is proven. 
    Conversely, if $tx_T$ is confirmed on blockchain $\mathcal{T}$, operators will observe $tx_T$ and subsequently confirm the corresponding transaction $tx_S$ on \textsf{vault$_S$}. 
    Therefore, $tx_S$ will be confirmed, a contradiction.
    
\end{packeditemize}
\end{proof}

\begin{theorem}[Atomicity]
Both $tx_S$ and $tx_T$ are either confirmed or not confirmed at all.
\end{theorem}
\begin{proof}
It follows immediately from Lemma~\ref{lemmap3g} and Lemma~\ref{lemmap4g}.
\end{proof}

\begin{theorem}[No online-client Requirement]
If the client is offline at any finite time $\tau$ after initiating a cross-chain transaction, the atomicity can ultimately be guaranteed.
\end{theorem}
\begin{proof}
The client initiates the cross-chain transaction by submitting $tx_S$ to the \textsf{vault$_S$}.
If the $tx_S$ fails, cross-chain transactions will not start at all, and both $tx_S$ and $tx_T$ are not confirmed.
If the $tx_S$ is executed on the blockchain $\mathcal{S}$, the $tx_S$ remains pending on the \textsf{vault$_S$}.
Then the execution of $tx_T$ falls into the following two cases.
In the first case, $tx_T$ is confirmed on the blockchain $\mathcal{T}$, $tx_S$ will also be confirmed by the operators, whether the client is online or not.
At this time, both $tx_S$ and $tx_T$ are confirmed
In the second case, $tx_T$ is failed. 
By Lemma~\ref{lemmap3g}, the $tx_S$ in \textsf{vault$_S$} will not be confirmed, and the client can challenge it at any time.
Thus, as the finite offline time $\tau$ pass, the client can initiate the challenge and cancel the $tx_S$, \ie, both $tx_S$ and $tx_T$ are not confirmed.
\end{proof}

\section{Evaluation}\label{sec:evaluate}
In this section, we evaluate \sysname in terms of on-chain verification costs and off-chain transaction throughput and latency.
We compare \sysname with three state-of-the-art counterparts, XClaim~\cite{zamyatin2019xclaim}, ZK-Bridge~\cite{xie2022zkbridge} and Tesseract~\cite{tesseract}. 
This section answers the following three questions: 
\begin{packeditemize}
    \item \textbf{Q1:} 
    What is the breakdown of on-chain costs in \sysname? 
    
    \item \textbf{Q2:} How does \sysname perform in terms of throughput and latency? 
    
    \item \textbf{Q3:} 
    How does \sysname compare with counterparts? 
\end{packeditemize} \
    

\subsection{System Implementation and Setup}\label{sec:imple}



We build a prototype of \sysname to evaluate its performance, employing the Ethereum test network Sepolia~\cite{sepolia}, which is consistent with the version of the Ethereum main network.
We use Solidity 0.8.0~\cite{solidity} to develop the smart contract (\ie,  $vaults$).
We deploy \sysname on AWS
EC2 machines with one m6a.xlarge instance per operator. 
All instances run with the Intel(R) Xeon(R) Platinum 8275CL CPU @ 3.00GHz and 16G RAM running Ubuntu Linux 22.04, with  AMD SEV as the TEE component~\cite{sev2020strengthening}.
The consensus protocol among operators relies on the implementation of RAFT~\cite{ongaro2014search}.
We conduct experiments in a LAN environment with 0.5 ± 0.03 ms and a WAN environment with 20 ± 0.1 ms, each machine was equipped with a private network interface with 100MB bandwidth. 

\begin{table}[t]
    \footnotesize
    \centering
    \setlength{\abovecaptionskip}{0cm}
    \caption{\textbf{Systerm Parameters.}} \label{table:parameter}
    \begin{tabular}{@{}ll@{}}
        \toprule[1pt]
        ~~~~~\textbf{Parameters}         & ~~~~~~~\textbf{Values}  \\
        \midrule
        Number of operators                       & 5, \textbf{10}, 15, 20       \\
        Batch size for RAFT consensus             & 500, 1000, 1500, \textbf{2000}      \\
        Number of transactions submitted on-chain           & 1, 10, \textbf{20}, 50, 100\\
        \bottomrule[1pt]
    \end{tabular}
\end{table}

We vary parameters in Table~\ref{table:parameter} with default values in bold to evaluate system gas cost, throughput, and latency. 
The experiment focuses on the cost on the chain and the throughput and delay of off-chain transaction processing.
Precisely, cost measures the gas cost for executing functions of the smart contract (\ie, \textsf{vault}), some of which can include multiple transactions simultaneously, thus reducing the shared cost for each transaction.
Throughput measures the number of transactions that operators can handle per second. 
Meanwhile, latency measures the time cost for operators to vote and process transactions.

\subsection{On-chain Cost of Functions} \label{sec:cost} 
Cross-chain exchange costs on the blockchain are measured in ``gas", which serves as a unit of measurement that gauges the computational effort needed to perform operations on Ethereum. To make these costs more understandable, we convert gas into USD with a gas price of 20 Gwei ($2\times10^{-8}$ ETH) and an ETH price of 2850 USD on 1 May 2024.


\bheading{Cost breakdown.} Table~\ref{table:evaluation} presents the on-chain costs of \sysname for each function involving 10 operators. 
Our measurements cover the entire process, including \textsc{Deposit}, \textsc{Transfer} (see~\ref{sec:normalcase}), and \textsc{Challenge} (see~\ref{sec:failures}) on Ethereum.
The \textsc{Transfer} function is particularly noteworthy as one of the most gas-consuming, approximately three times that of the Deposit method (84K gas). 
This higher cost is attributed to the verification of the multi-signature of TEEs, which demands a substantial amount of computation. 
However, with a batch size of 20 transactions, the per-transaction fee is amortizing to 14K gas (0.82 USD), less than a simple ETH transfer\footnote{Sample ETH transfer is an example of sending ETH in Ethereum, commonly employed for comparing gas fees~\cite{pose}.\label{footnote2}}.
A successful challenge-response involves both the \textsc{StartChallenge} and \textsc{ResolveChallenge} functions, totaling 85K gas (4.17 USD). The \textsc{UpdateCheckpoint} function consumes 260K gas. 
However, as we implement a checkpoint update every 20 transactions, the cost per transaction for \textsc{UpdateCheckpoint} is 1.48 USD.
Therefore, the total cost of a cross-chain transaction, including deposit, transfer and update checkpoint, is 125K gas per transaction.

\begin{table}[t]
\centering
    \begin{threeparttable}
    \caption{\textbf{On-chain cost.} The costs are in USD as per exchange rates of 1 May 2024: ETH/USD 2850 and gas price is 20 Gwei ($2\times10^{-8}$ ETH).}\label{table:evaluation}
    
    \begin{tabular}{@{}lrrc@{}}
        \toprule
        Functions           & Total (gas)  & Average (gas/tx)      & USD (\$) \\
        \midrule
        \textsc{Deposit}           & 84,404~~~~       &84404~~~~~~             & 4.81       \\
        \textsc{Transfer}          & 289,018~~~~      &14451~~~~~~             & 0.82        \\
        \textsc{StartChallenge}    & 47,351~~~~       &47351~~~~~~             & 2.70       \\
        \textsc{ResolveChallenge}  & 37,364~~~~       &37364~~~~~~             & 2.13       \\
        \textsc{UpadateCheckpoint} & 259814~~~~      &25981~~~~~~              & 1.48       \\
        \midrule
        Simple ETH transfer & 21,000~~~~ &21000~~~~~~         & 1.20      \\
        \bottomrule
    \end{tabular}
\end{threeparttable}
\end{table}

 \begin{figure}[t]
    \centering
    \includegraphics[width=8cm]{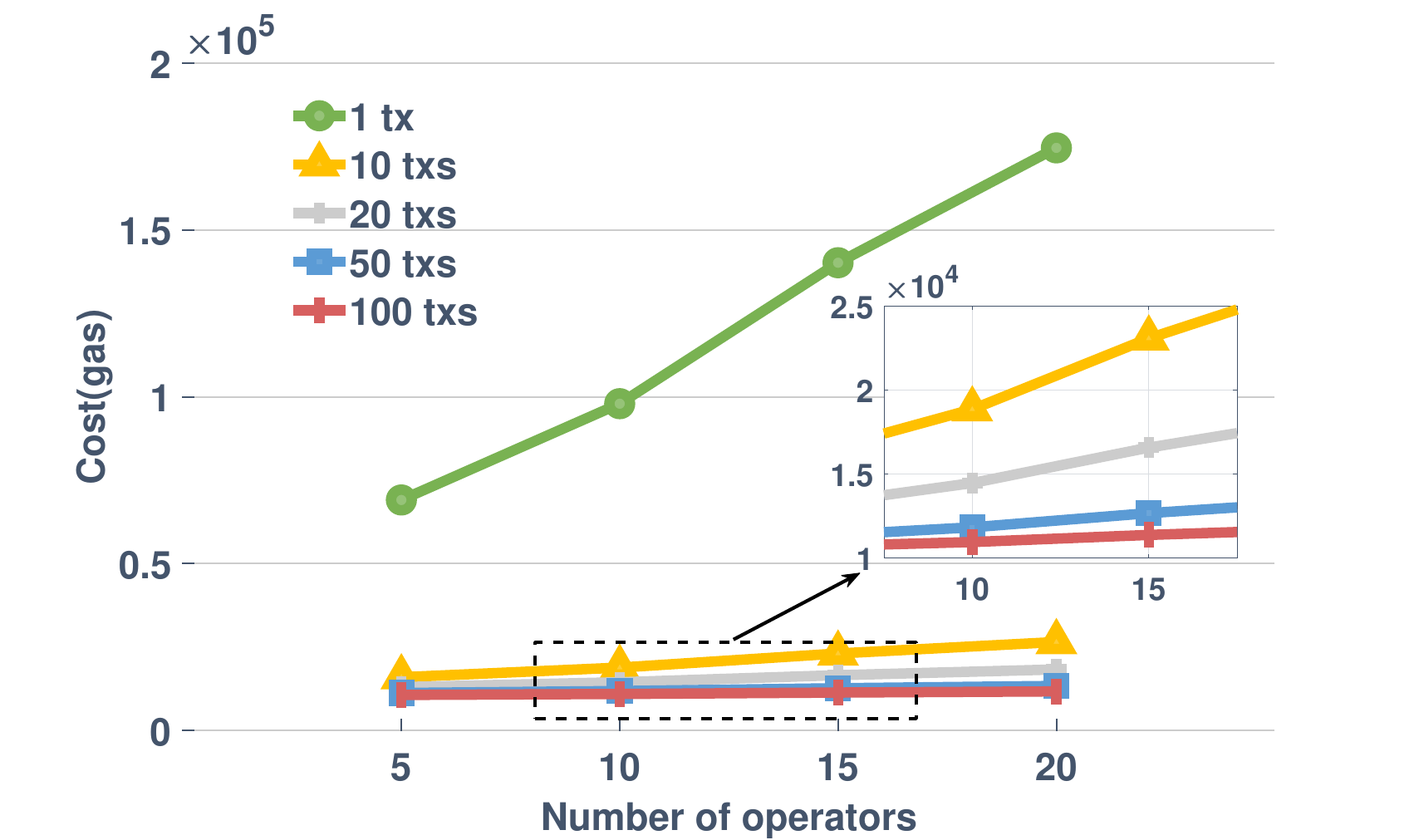}
    \caption{On-chain cost of the \textsc{Transfer} function with varying numbers of operators and transactions.}
    \label{fig:cost}
    \end{figure} 



\bheading{Cost of the \textsc{Transfer} function.}  
As mentioned in~\secref{sec:design}, \textsc{Transfer} function submits more than one transaction at a time.
Thus, we provide detailed insights into the average gas requirements for the \textsc{Transfer} function, which is the most gas-consuming function in \sysname. 
\figref{fig:cost} illustrates transfer costs in MERCURY across varying numbers of operators and batch sizes.
The observed trend indicates that the cost of \textsc{Transfer} grows with an increasing number of operators, reaching 175K gas when there are 20 operators without batching (\ie, 1 transaction per batch). 
This substantial cost underscores the necessity of implementing batching to reduce the per-transaction expenses.
As the batch size increases, there is a significant reduction in average costs. 
As the batch size further increases to 100 transactions, the cost diminishes to just 12K gas (0.67 USD).

\subsection{Performance Evaluation}\label{sec:performance}
We evaluate two performance metrics: 1) throughput, measured in thousands of transactions per second (KTPS),  representing the number of transactions executed per second, and 2) latency, measured in millisecond (ms), denoting the average end-to-end delay from the moment operators get the transactions until the submission of the transaction.



\bheading{Throughput and latency.}
\figref{fig:performance} illustrates the throughput and latency in \sysname across varying numbers of operators and batch sizes in LAN and WAN environments. 
In the LAN environment, throughput shows a modest decrease while latency increases as the number of operators rises, peaking at 38 ms with a batch size of 1000 and 20 operators. This occurs because the increase in the number of operators leads to a rise in messaging overhead, causing bandwidth inefficiencies.
Throughput exhibits an upward trend with increasing batch size, while latency decreases, reaching 10 ms at a batch size of 1000 with 5 operators. This is because, with larger batch sizes, \sysname can process more transactions concurrently.

In the WAN environment, throughput shows a modest decrease while latency increases as the number of operators rises. 
This trend is more pronounced than in the LAN environment because the increased communication delay in WAN amplifies the impact of adding more operators on system performance. 
Additionally, both throughput and latency increase with larger transaction batch sizes.


\bheading{TEE overhead.} 
We evaluate the throughput and latency of \sysname with and without SEV across 5-20 operators in both LAN and WAN environments.
In the LAN environment, the introduction of SEV results in a halving of throughput and nearly doubles the latency. 
This degradation occurs because operations executed inside a TEE require encryption and context switching from the host machine’s regular execution environment, which hampers performance. 
In the WAN environment, the introduction of SEV also leads to a decline in throughput and an increase in latency, although the trend is less pronounced. 
In this case, the inherent communication delay of the WAN environment becomes the primary factor affecting throughput and latency.



\begin{figure}[t]
	\centering
	\begin{subfigure}{0.48\linewidth}
		\centering
		\includegraphics[width=1.1\linewidth]{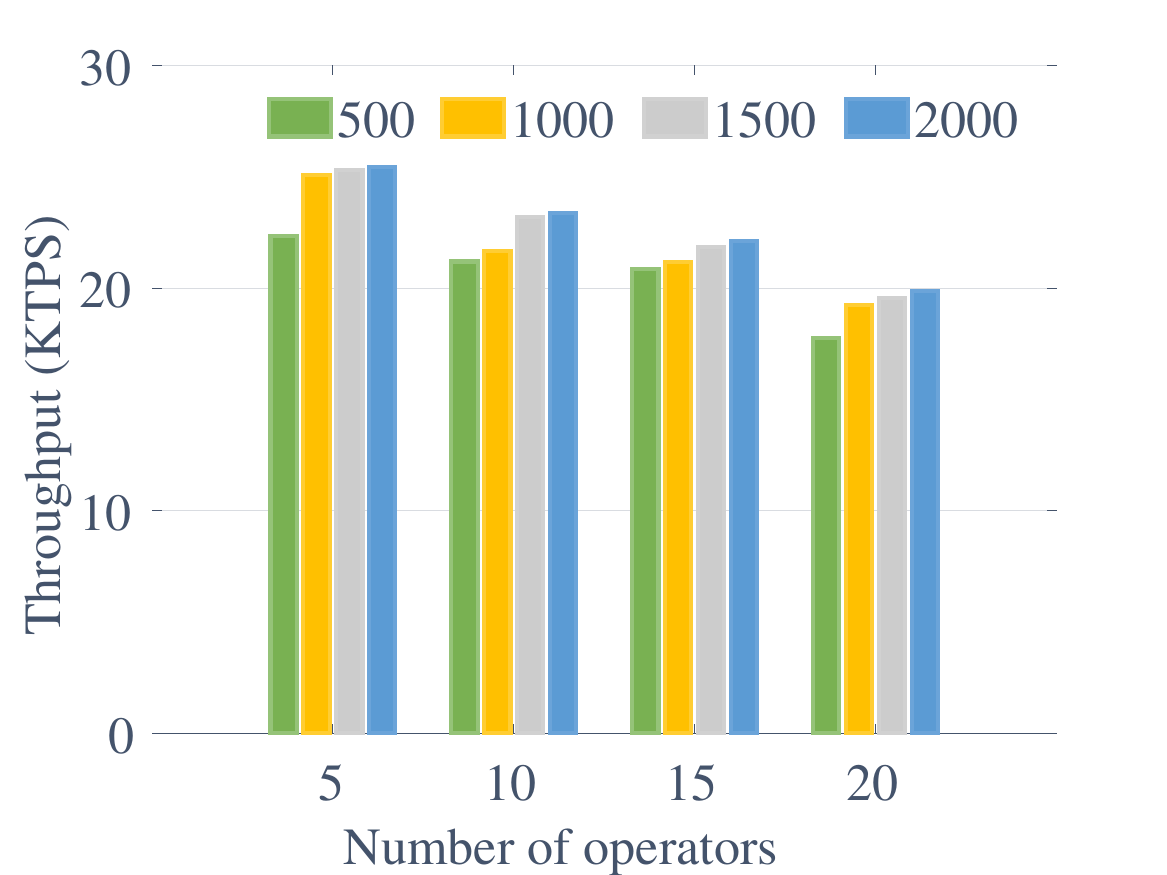}
		\caption{Throughput in LAN}
		\label{fig:throughputLAN}
	\end{subfigure}
	\hfill
	\begin{subfigure}{0.48\linewidth}
		\centering
		\includegraphics[width=1.1\linewidth]{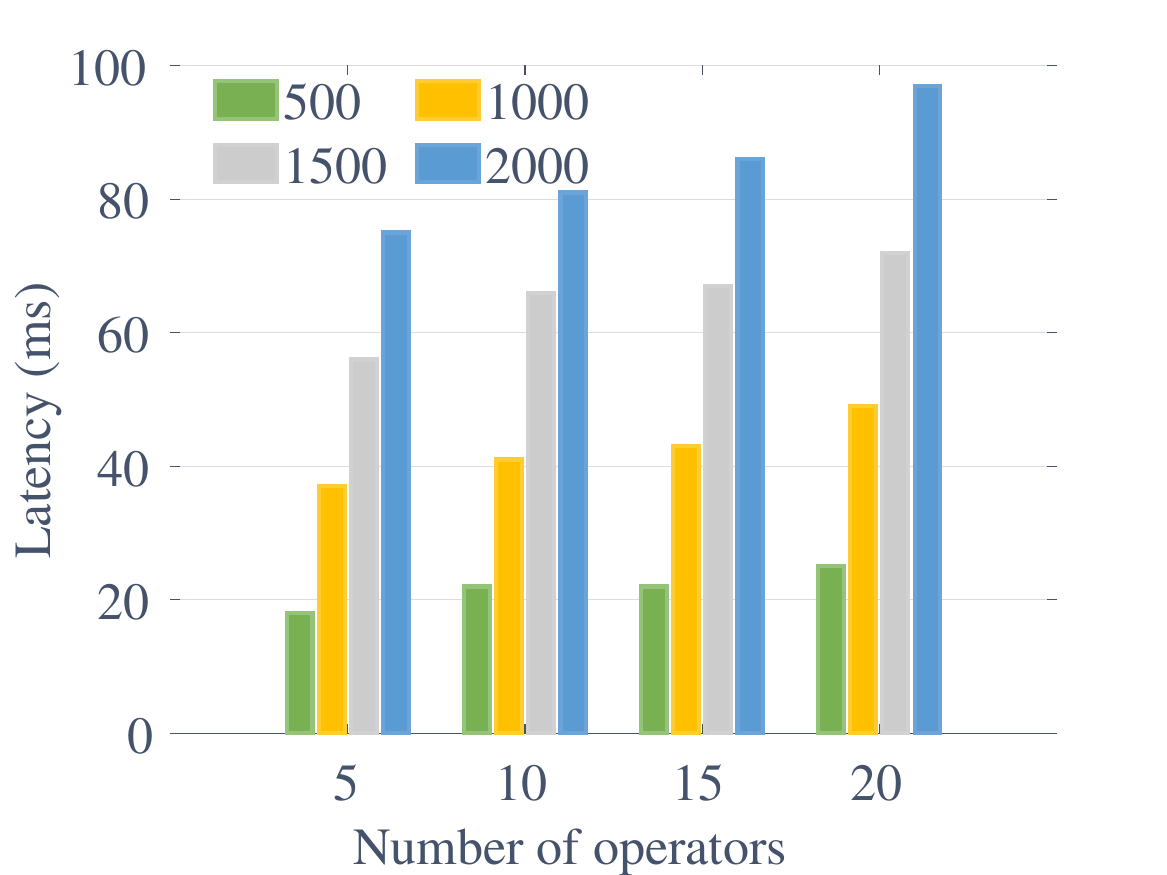}
		\caption{Latency in LAN}
		\label{fig:latencyLAN}
	\end{subfigure}
        \hfill
	\begin{subfigure}{0.48\linewidth}
		\centering
		\includegraphics[width=1.1\linewidth]{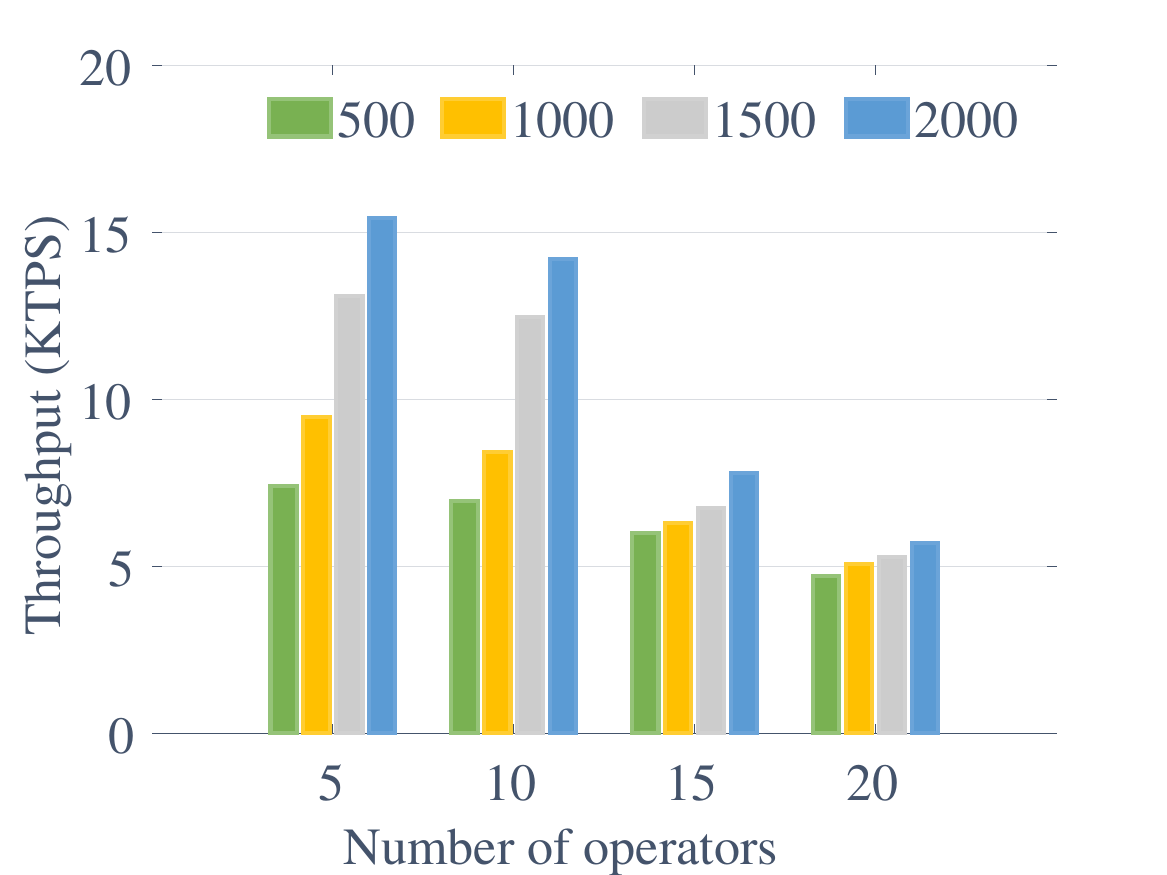}
		\caption{Throughput in WAN}
		\label{fig:throughputWAN}
	\end{subfigure}
	\hfill
	\begin{subfigure}{0.48\linewidth}
		\centering
		\includegraphics[width=1.1\linewidth]{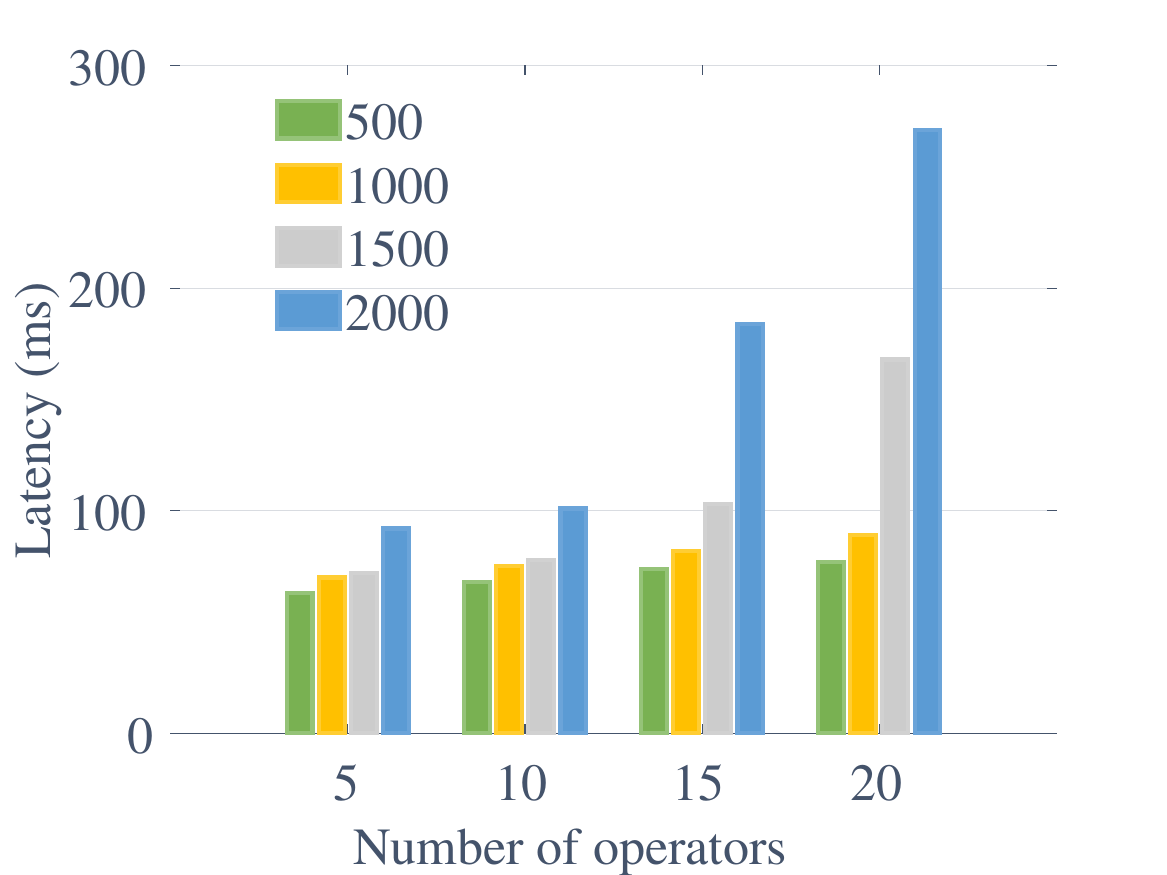}
		\caption{Latency in WAN}
		\label{fig:latencyWAN}
	\end{subfigure}
	\caption{Transaction processing performance with varying different numbers of operators and batch size in LAN and WAN.}
	\label{fig:performance}
\end{figure}

\begin{figure}[t]
	\centering
	\begin{subfigure}{0.49\linewidth}
		\centering
		\includegraphics[width=1.1\linewidth]{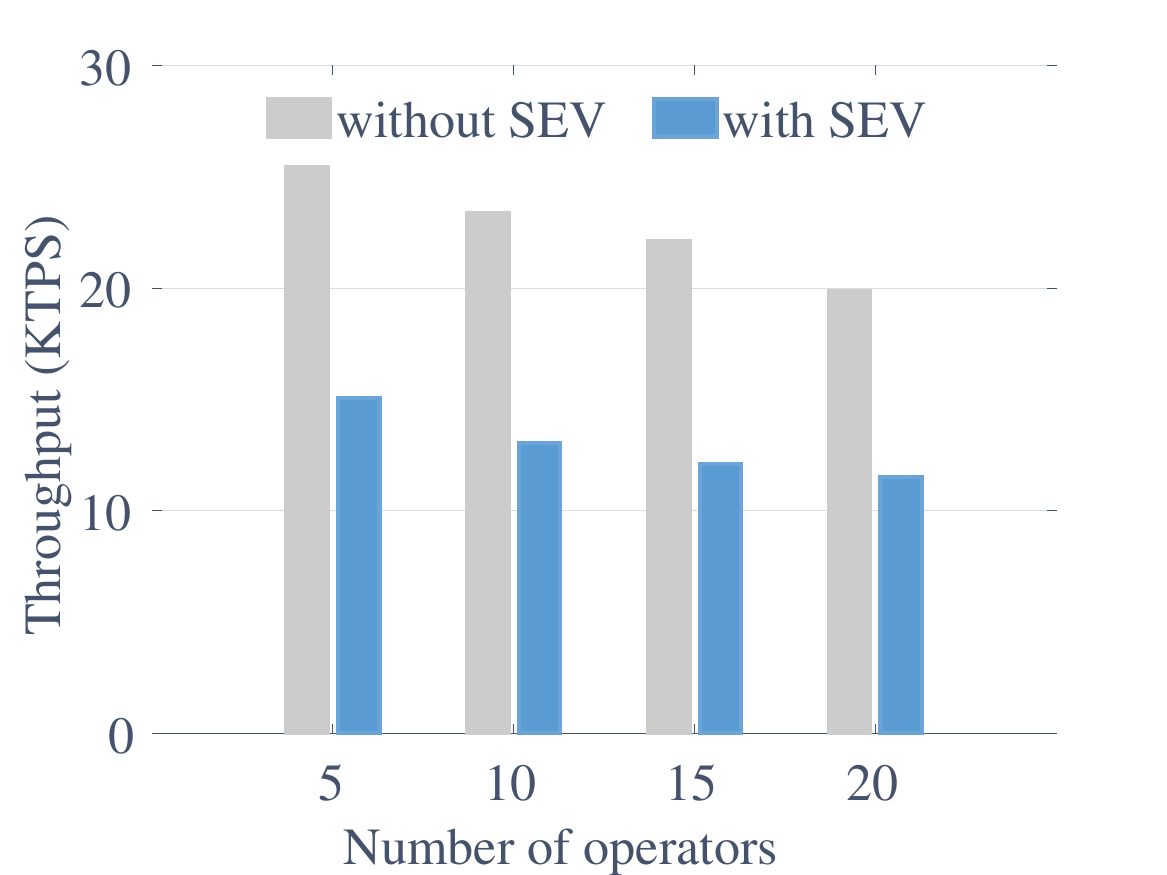}
		\caption{Throughput in LAN}
		\label{fig:throughputnodeLAN}
	\end{subfigure}
	\hfill
	\begin{subfigure}{0.49\linewidth}
		\centering
		\includegraphics[width=1.1\linewidth]{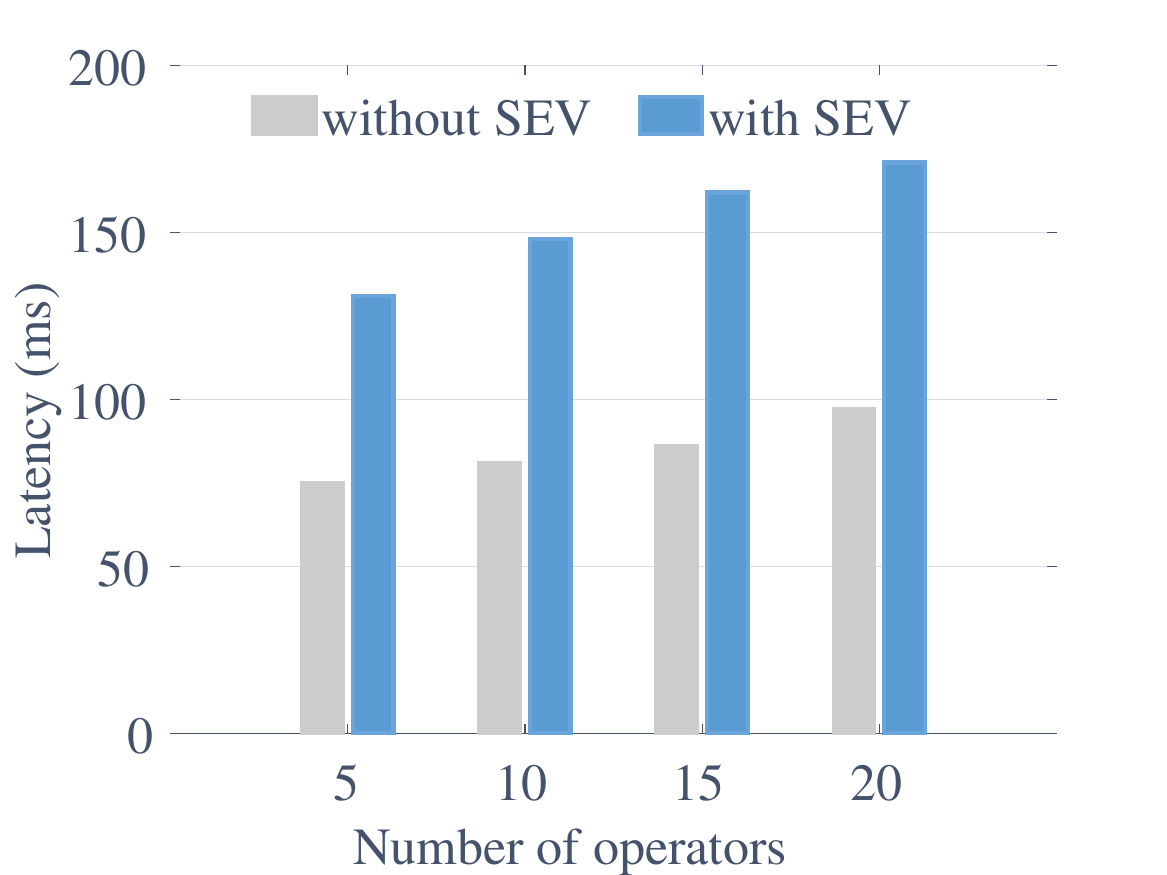}
		\caption{Latency in LAN}
		\label{fig:lantencynodeLAN}
	\end{subfigure}
	\hfill
	\begin{subfigure}{0.49\linewidth}
		\centering
		\includegraphics[width=1.1\linewidth]{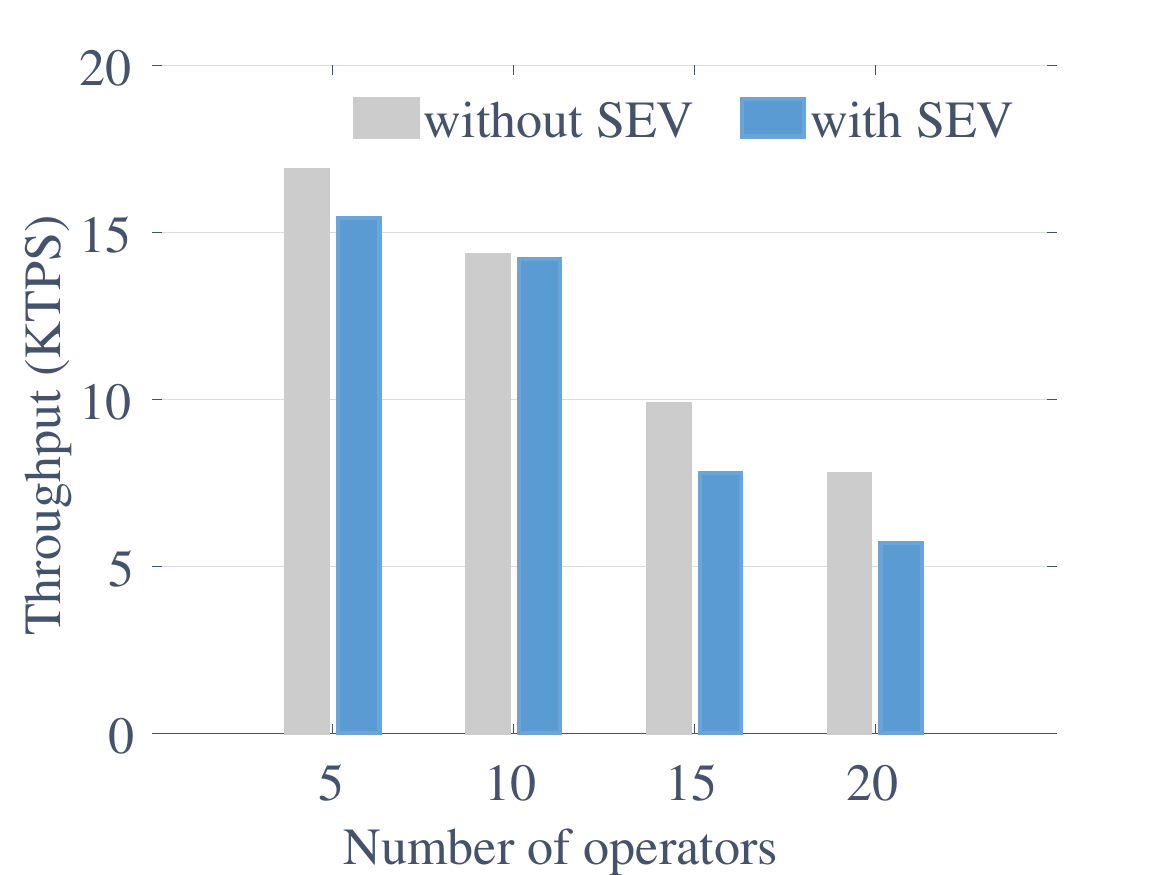}
		\caption{Throughput in WAN}
		\label{fig:throughputnodeWAN}
	\end{subfigure}
	\hfill
	\begin{subfigure}{0.49\linewidth}
		\centering
		\includegraphics[width=1.1\linewidth]{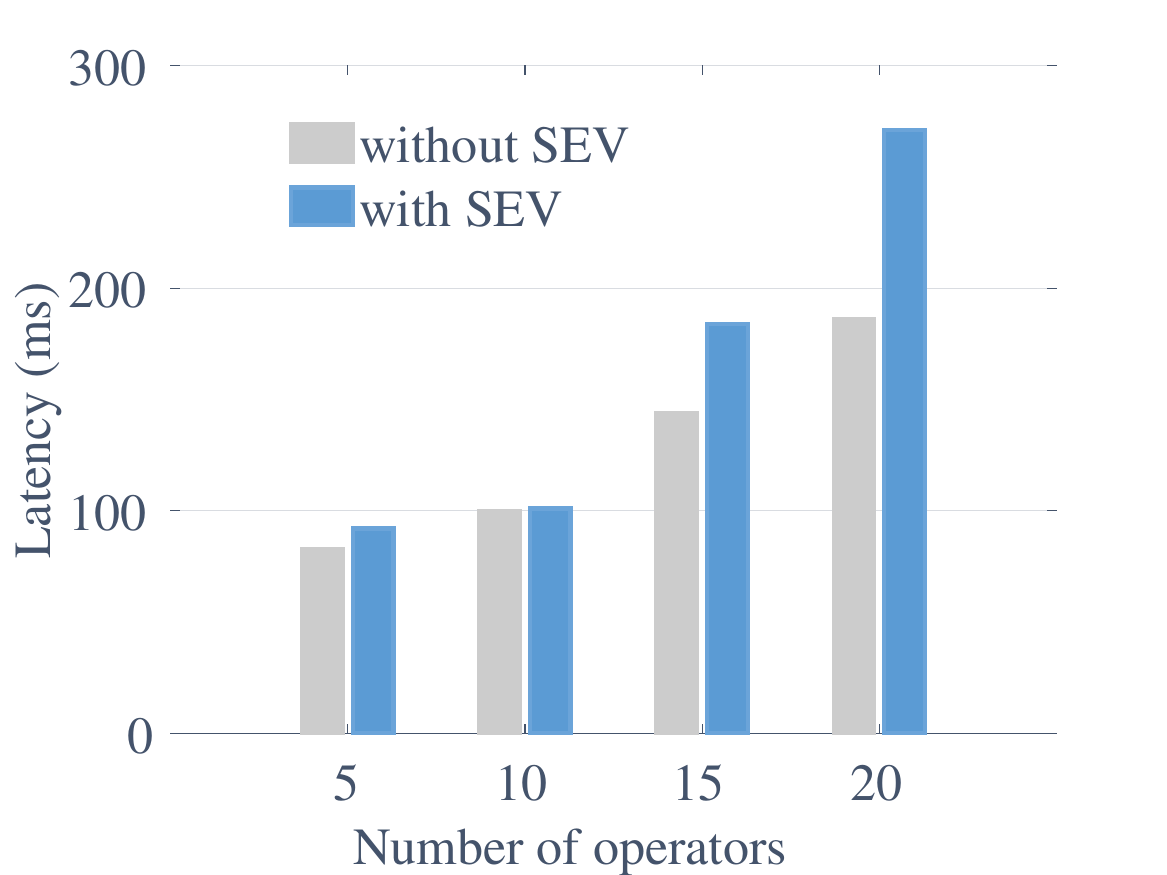}
		\caption{Latency in WAN}
		\label{fig:lantencynodeWAN}
	\end{subfigure}
	\caption{System performance with/without SEV in LAN and WAN.}
	\label{fig:sev}
\end{figure}

\subsection{Comparison with Counterparts}\label{sec:compare}

    \begin{table}[t]
    \footnotesize
    \centering
    \setlength{\abovecaptionskip}{0cm}
    \caption{\textbf{Comparison with state-of-the-art counterparts.}} \label{table:comparecost}
    \begin{tabular}{@{}lllll@{}}
        \toprule[1pt]
                               & XClaim   &  ZK-Bridge    & Tesseract & \sysname  \\
        \midrule
         Cost (gas/tx)    & $>$388882   & $>$227000 & 239181 & 124836     \\
          USD (\$)            & $>$22.17   & $>$12.9 & 13.63 & 7.11     \\
        \bottomrule[1pt]
    \end{tabular}
    \vspace{-0.4cm}
\end{table}

\bheading{Comparison with XClaim.}
A cross-chain exchange for XClaim~\cite{zamyatin2019xclaim} consists of the locking of the currencies on the backed chain (\ie, source blockchain), the issuing of cryptocurrency-backed assets (CBAs) on the issuing chain (\ie, target blockchain). 
As shown in Table~\ref{table:evaluation}, the cost for processing one cross-chain exchange in \sysname is about 125K gas.
 In contrast, the on-chain cost of XClaim for a cross-chain exchange is about 389K gas (Table~\ref{table:comparecost}), representing a 67.87\% increase compared to \sysname.
Besides, to synchronize the state of the backed chain to the issuing chain,  XClaim incurs a daily cost of 28380K gas (1401.97 USD) to relay blocks, which is a significant additional expense that cannot be overlooked.


\bheading{Comparison with ZK-Bridge .}
ZK-Bridge~\cite{xie2022zkbridge} utilizes zero-knowledge proofs that are generated off-chain and verified on-chain. 
In contrast, \sysname employs off-chain transaction processing and verifies multi-signatures on-chain. 
We specifically focus on two fundamental steps common to both protocols: the generation of proofs/multi-signatures, and the on-chain verification. 
ZK-Bridge's transaction proof generation takes approximately 18s, while, as depicted in~\figref{fig:performance}, \sysname's latency is less than 300ms, showcasing significant efficiency improvements.
The gas consumption for on-chain verification of multi-signatures in \sysname is illustrated in~\figref{fig:cost}, which is comparable with the 227K gas (Table~\ref{table:comparecost}) required for ZK-Bridge verification.


\bheading{Comparison with Tessercat.} {Tesseract~\cite{tesseract}, a TEE solution for cross-chain exchange, generates transaction pairs with the hash time lock. 
For a fair comparison, we deployed the HTLC on Ethereum to measure the on-chain cost~\cite{ethhtlc}.}
A cross-chain transaction with Tesseract involves two phases: lock and withdraw. 
The on-chain costs for the two phases are 165K gas and 74K gas, respectively, resulting in a total cost of 239K gas (13.63 USD). 
By contrast, \sysname's cost is 125K gas, which declines by 47.70\%. 

\section{Beyond \sysname}
In this paper, \sysname is built atop blockchains that support smart contracts. 
However, it can be easily extended to blockchains (\eg, Bitcoin~\cite{nakamoto2008bitcoin}) that do not support smart contracts. 
The main extension is to replace the challenge-response mechanism with hash time-lock to resolve redemption issues caused by TEEs' unavailability. We now introduce it in detail.

\bheading{Hash time-lock.}
A hash time-lock is a transactional agreement used in the cryptocurrency industry to facilitate conditional payments. It incorporates both a hashlock and a timelock.
The hash lock ensures that the transaction includes a hashed secret, and it can only be successfully executed by revealing this secret (i.e., the preimage of the hash).
The time lock sets a time limit $T$, specifying that the transaction can only be executed after this time has passed.

\bheading{Contract Design.}
We consider the case, in which a client who expects to exchange currency $s$ on blockchain $\mathcal{S}$ for currency $t$ on blockchain $\mathcal{T}$.
First, the client sends an exchange request to the operators, who then generate a pair of transactions $tx_S$ and $tx_T$, while the preimage is stored inside the TEE. 
The execution of the transaction $tx_S$ falls into two cases.
\begin{packeditemize}
    \item $tx_S^1$: If the time exceeds $T$, the client's money will be returned to the original address.
    \item $tx_S^2$: If a secret preimage is revealed, the client's money will be transferred to an address controlled by \sysname.
\end{packeditemize}

The client signs the transaction $tx_S$ and publishes it on the blockchain $\mathcal{S}$. 
If, within time T, the operator detects that transaction $tx_T$ has been successfully executed, it will reveal the hash preimage on blockchain $\mathcal{S}$, resulting in the successful execution of transaction $tx_S^2$.
Otherwise, after time T, transaction $tx_S^2$ will be successfully executed, and the client will receive a refund on blockchain $\mathcal{S}$.
As we can see, the hash time-lock ensures the atomicity of the cross-chain transaction when the operators are unavailable.

\section{Conclusion and Future Work}\label{sec:conclusion}
In this paper, we propose \sysname, a trust-minimized and 
efficient cross-chain exchange without online-client requirement. 
\sysname employs a group of operators protected by TEEs as the trust root.
The cost-efficient challenge-response mechanism safeguards the atomicity of cross-chain exchanges, even when all TEEs are unavailable.
Furthermore, \sysname employs a lightweight transaction verification mechanism and simultaneously incorporates multiple optimizations to reduce on-chain costs.
The prototype implementation of  \sysname delivers promising results, surpassing existing state-of-the-art approaches in terms of on-chain cost and off-chain processing time. 
There are several future directions to improve \sysname. First, 
\sysname can be extended to support Proof-of-Work (PoW)-based blockchains like Bitcoin that do not have smart contracts. 
Second, more financial functions can be built on top of \sysname, \eg, lending and borrowing.



\normalem
\bibliographystyle{IEEEtran}
\bibliography{main}

\end{document}

%% file: macro.tex
\newcommand{\bheading}[1]{{\vspace{4pt}\noindent{\textbf{#1}}}}

\newcounter{note}[section]
\renewcommand{\thenote}{\thesection.\arabic{note}}
\newcommand{\niu}[1]{\refstepcounter{note}{\bf\textcolor{red}{$\ll$JY~\thenote: {\sf #1}$\gg$}}}

\newcommand{\blackding}[1]{\ding{\numexpr181+#1\relax}}

\newcommand{\secref}[1]{\mbox{Sec.~\ref{#1}}\xspace}

\newcommand{\figref}[1]{\mbox{Fig.~\ref{#1}}}

\newcommand{\ignore}[1]{}

\newcommand{\ie}{\textit{i.e.}\xspace}
\newcommand{\eg}{\textit{e.g.}\xspace}

\newcommand{\etal}{\textit{et al.}\xspace}

\newcommand{\sysname}{\textsc{Mercury}\xspace}

\newcounter{packednmbr}

\newenvironment{packeditemize}{
\begin{list}{$\bullet$}{
\setlength{\labelwidth}{0pt}
\setlength{\itemsep}{2pt}
\setlength{\leftmargin}{\labelwidth}
\addtolength{\leftmargin}{\labelsep}
\setlength{\parindent}{0pt}
\setlength{\listparindent}{\parindent}
\setlength{\parsep}{1pt}
\setlength{\topsep}{1pt}}}{\end{list}}

\newtheorem{theorem}{Theorem}

\newtheorem{lemma}{Lemma}